\documentclass[lettersize,journal]{IEEEtran}
\IEEEoverridecommandlockouts

\usepackage{amsthm}
\usepackage{mathrsfs}
\usepackage{amssymb}
\usepackage[cmex10]{amsmath}
\usepackage{stfloats}
\usepackage{graphicx}
\usepackage{subfigure}
\usepackage{tabularx}
\usepackage{epsfig,epsf,color,balance,cite}
\usepackage{verbatim}
\usepackage{url}
\usepackage{bm}
\usepackage{caption}
\usepackage{threeparttable}
\usepackage{diagbox}
\usepackage{multirow}
\usepackage{stfloats}
\usepackage[ruled,linesnumbered]{algorithm2e}

\newtheorem{theorem}{Theorem}
\newtheorem{lemma}[theorem]{Lemma}
\newtheorem{proposition}[theorem]{\bf Proposition}
\newtheorem{coro}[theorem]{Corollary}
\makeatletter

\newcommand{\Rmnum}[1]{\expandafter\@slowromancap\romannumeral #1@}
\makeatother
\usepackage{amsmath}
\definecolor{myc1}{rgb}{0,0,1}
\allowdisplaybreaks[1]

\begin{document}
	\title{{Joint User Scheduling and Computing Resource Allocation Optimization in Asynchronous Mobile Edge Computing Networks}}
	
	     \author{
		Yihan Cang, Ming Chen, Yijin Pan, Zhaohui Yang, \\Ye Hu, Haijian Sun, and Mingzhe Chen
		\vspace{-3em}
		\thanks{The work of Ming Chen was supported by the National Natural Science Foundation of China (NSFC) under grant 61871128, 61960206005 and 61960206006, by the National Key Research and Development Program under grant 2018YFB1801905, and by Fundamental Research on Foreword Leading Technology of Jiangsu Province under grant BK20192002. 
		The work of Y. Pan was supported by Chongqing Natural Science Joint Fund Project under Grant No. CSTB2023NSCQ-LZX0121. (Corresponding author: Ming Chen).} 
		\IEEEcompsocitemizethanks{
			\IEEEcompsocthanksitem Y. Cang, Ming Chen and Y. Pan  are with the National Mobile Communications Research Laboratory, Southeast University, Nanjing 210096, China (e-mails: yhcang@seu.edu.cn, chenming@seu.edu.cn,  panyj@seu.edu.cn).  Ming Chen is also with the Purple Mountain Laboratories, Nanjing 211100, China.
		\IEEEcompsocthanksitem
			Z. Yang is with College of Information Science and Electronic Engineering, Zhejiang University, Hangzhou 310027, China, and with International Joint Innovation Center, Zhejiang University, Haining 314400, China, and also with Zhejiang Provincial Key Laboratory of Info. Proc., Commun. \& Netw. (IPCAN), Hangzhou 310027, China (e-mail: yang\_zhaohui@zju.edu.cn).
   \IEEEcompsocthanksitem 
            Y. Hu is with the Department of Industrial and System Engineering, University of Miami, Coral Gables, FL, 33146, (e-mail: yehu@miami.edu).
		\IEEEcompsocthanksitem
		H. Sun is with the School of Electrical and Computer Engineering,
		The University of Georgia, Athens, GA 30602 USA (e-mail: hsun@uga.edu). 
  \IEEEcompsocthanksitem
			M. Chen is with the Department of Electrical and Computer Engineering and Institute for Data Science and Computing, University of Miami, Coral 			Gables, FL, 33146, USA (e-mail: mingzhe.chen@miami.edu). 
	}
	}

     \maketitle
	\begin{abstract}
    In this paper, the problem of joint user scheduling and computing resource allocation in asynchronous mobile 
    edge computing (MEC) networks is studied. In such networks, edge devices will offload their computational tasks to an MEC server, using the energy they harvest from this server. 
    To get their tasks processed on time using the harvested energy, edge devices will strategically schedule their task offloading, and compete for the computational resource at the MEC server.
    Then, the MEC server will execute these tasks asynchronously based on the arrival of the tasks. 
    This joint user scheduling, time and computation resource allocation problem is posed as an optimization framework whose goal is to find the optimal scheduling and allocation strategy that minimizes the energy consumption of these mobile computing tasks.
    To solve this mixed-integer non-linear programming problem, the general benders decomposition method is adopted which decomposes the original problem into a primal problem and a master problem. Specifically, the primal problem is related to computation resource and time slot allocation, of which the optimal closed-form solution is obtained. The master problem regarding discrete user scheduling variables is constructed by adding optimality cuts or feasibility cuts according to whether the primal problem is feasible, which is a standard mixed-integer linear programming problem and can be efficiently solved. By iteratively solving the primal problem and master problem, the optimal scheduling and resource allocation scheme is obtained. 
    Simulation results demonstrate that the proposed asynchronous computing framework reduces $87.17\%$ energy consumption 
    compared with conventional synchronous computing counterpart.
	\end{abstract}
	
	\begin{IEEEkeywords}
		Mobile edge computing, asynchronous computing, user scheduling, wireless power transfer.
	\end{IEEEkeywords}
	\section{Introduction}

Mobile edge computing (MEC) provides powerful computing ability to edge devices \cite{9113305,8016573}. Numerous works have investigated MEC systems from the perspective of resource allocation. In \cite{9756368}, computing offloading and service caching are jointly optimized in MEC-enabled smart grid to minimize system cost. Work \cite{10003098} proposed a reverse auction-based offloading and resource allocation scheme in MEC. With the aid of machine learning, a multi-agent deep deterministic policy gradient
(MADDPG) algorithm is designed to  maximize energy efficiency in \cite{9832640}.  However, deploying computing resources at edge servers of a wireless network faces several challenges. First, due to limited energy of edge servers, they may not be able to provide sufficient computation resource according to devices' requirements \cite{8488502,8543183}.  Second, executing all offloading tasks synchronously requires edge servers to wait the arrival of the task with maximum transmission delay which may not be efficient. Meanwhile, task scheduling sequence is nonnegligible in synchronous task offloading, which will also impact the network loads and task completion \cite{9580352}.

	
To address the first challenge, wireless power transfer (WPT)  technology that exploits energy carried by  radio frequency (RF) signals  emerges \cite{lu2014wireless}. 
Instead of using solar and wind sources, ambient RF signals 
can be a viable new source for energy scavenging. 
 Harvesting energy from the environment provides perpetual energy supplies to wireless 
devices for tasks offloading \cite{8357386}. Thus, WPT has been regarded as a promising paradigm for MEC 
scenarios. 
Combining WPT with MEC, 
the authors in \cite{8234686} proposed a multi-user wireless-powered MEC framework aiming at minimizing the total energy consumption under latency constraints. In \cite{8334188}, considering binary computation offloading, the weighted sum computation rate of all wireless devices was maximized by optimizing computation mode selection and transmission time allocation.  The work in \cite{9881553} proposed a multiple intelligent reflecting surfaces (IRSs) assisted wireless powered MEC system, where the IRSs are deployed to assist both the downlink WPT from the  access point (AP) to the wireless devices and the uplink computation offloading. 
However, the above works \cite{8264794,7997360,8249785,8986845,8234686,8334188,9881553} assumed that all computational tasks offloaded by users will arrive at the server at the same time and then the server starts to process all tasks simultaneously,  which is not efficient and even impractical due to users' dynamic computational task processing requests 
 \cite{xu2020task}.

Currently, only a few existing works \cite{xu2020task,dai2019delay,8468240,9924248,9292432,9201170} optimized MEC networks under dynamic computation requests. 
The work in \cite{xu2020task} designed a Whittle index based offloading algorithm to maximize the long-term reward for asynchronous MEC networks where computational tasks arrive randomly. 
In  \cite{dai2019delay}, the authors studied the co-channel interference caused by asynchronous task uploading in  NOMA MEC systems. The work in \cite{8468240} investigated the energy efficient asynchronous task offloading for a MEC system where computational tasks with various latency requirements arrive at different time slots. 
Task scheduling problem for MEC systems with task interruptions and insertions was studied in \cite{9924248}.
However, the above works \cite{xu2020task,dai2019delay,8468240,9924248} that focused on the asynchronous task offloading neglected how the asynchronous task arrival affects the computation at the MEC server. 
	The authors in \cite{9292432} used a sequential computation method to solve the energy consumption minimization problem under asynchronous task arrivals. The work in \cite{9201170} designed a computation strategy that only allows a task to be executed after the completion of the previous tasks. Yet, works in \cite{9292432} and \cite{9201170} are still constrained by their limited usage of the server computation capacity, and cannot act as resource efficient asynchronous task offloading solutions. 
 
 The sequential computation strategy \cite{8692421} has shown to have the potential to improve the computation resource efficiency and  task execution punctuality in an asynchronous MEC network.  
However, since the computation resource allocation at the server depends on the arrival of the offloaded tasks, the sequential scheduling of the tasks will inevitably affect system performance, which is a fact that has been wildly ignored \cite{9292432,8952620,9632276}.

The main contribution of this paper is a novel asynchronous MEC framework that jointly schedules tasks and allocates computation resource with optimized system energy efficiency. 
In brief, our key contributions include:
	\begin{itemize}
	
	\item 
 We develop a novel framework to manage computation resource for the sequential computation in  asynchronous MEC networks. In particular, we consider a MEC network in which the edge devices sequentially harvest energy for transmission,  offload their computational tasks to a MEC server, and then compete for computation source at the server to get their tasks accomplished. To achieve the high energy efficient task execution, a policy
needs to be designed for determining the optimal task scheduling sequence, time  and computational  resource allocation. We pose this
joint scheduling and resource allocation problem in an optimization framework and seek to find the strategy which minimizes the energy consumption of the tasks.

	\item Then, a general benders decomposition (GBD) based algorithm is proposed to solve the formulated mixed-integer non-linear programming (MINLP) problem  which is decomposed into a primal problem that allocates computation resource and time, and a master problem that schedules user tasks. 
By iteratively solving  the primal problem and master problem, the optimal scheduling and resource allocation scheme is obtained.

	\item To show the effectiveness of the proposed algorithm, we prove that the optimal energy efficient scheduling and resource allocation scheme also optimizes the task punctuality. Our analytical results also show that the optimal allocation scheme for a given offloading task follows a specific pattern: the computation frequency allocated to each task remains constant initially, then gradually decreases before eventually reaching zero. Notably, all tasks experience a simultaneous decrease, the time of which is given in a closed form, in terms of their required central processing unit (CPU) cycles. 
 Leveraging these identified properties, we introduce a computation resource allocation algorithm that offers a low-complexity solution.

\end{itemize} 
	Simulation results demonstrate
that the proposed asynchronous computing framework reduces
$87.87\%$ energy consumption 
compared with
conventional synchronous computing counterpart. Moreover, computational  complexity of the proposed computation resource allocation algorithm is reduced by $100$ times compared with conventional interior point method.


The rest of this paper is structured as follows. Section II elaborates  system model and problem formulation. In Section III, we investigate the properties of asynchronous frequency allocation with given time allocation and user scheduling. The joint optimization of user scheduling, time allocation, and computation resource  allocation  is rendered in Section IV. Simulation results are presented in Section~V. Finally, Section VI draws the conclusions.

	\section{System Model and Problem Formulation}
		\begin{figure}[t]
		\centering
		\includegraphics[width=0.5\textwidth]{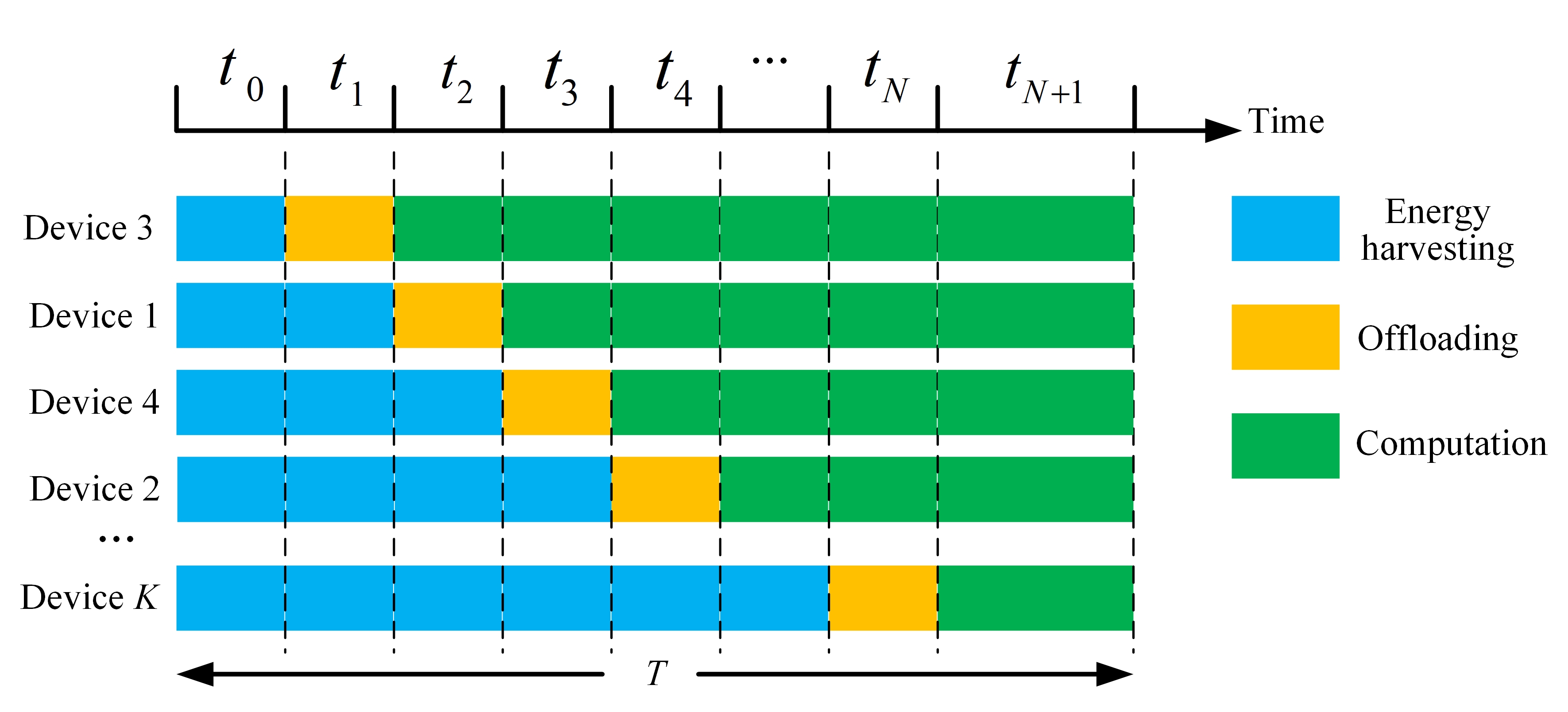}
		\caption{An illustration about the flow chart of the ordered TDMA system with asynchronous computing.} \label{time}
		\vspace{-1.em}
	\end{figure}

 Consider a MEC network consisting of one MEC server, and a set $\mathcal{K}=\{1,2,\cdots,K\}$ of energy harvesting enabled edge devices.  Within this network, each device $k\in\mathcal{K}$ needs to execute an $A_k$ bits computational task, and will offload its  computational task to the
 MEC server. As shown in Fig.~\ref{time}, the devices need to first harvest energy from the server to enable such offloading. Then, using the time division multi-access (TDMA) technique, the devices need to schedule their offloading toward the MEC server. In other words, the computational tasks offloaded by devices will arrive at the MEC server asynchronously. To this end, the MEC server will process each device’s computational task in an asynchronous manner. In particular, the server will process devices’ computational tasks according to the time that it receives each of these computational tasks. 
 




The server and edge devices must complete their computational tasks within a time period $T$ which is divided into $\left(K+2\right)$ time slots. The duration of each time slot $n\in\{0,\cdots,N+1\}$ is  represented by $\Delta t_n$, with $N=K$.  
Each device uses one time slot to offload its computational task. 
  Let $a_{k,n}$ be the index to indicate whether device $k$ offloads its task to the server at 
 time slot $t_n$. In particular, if device $k$ uses time slot $t_n$ to offload its computational task, we have $a_{k,n}=1$; otherwise, $a_{k,n}=0$. 
Since each device uses only one time slot and each time slot can only be allocated to one device, we have 
 $\sum_{k=1}^K a_{k,n}=1,\  \forall n\in\{1,\cdots,N\}$, and $\sum_{n=1}^N a_{k,n}=1,\  \forall k\in\mathcal{K}$.
Meanwhile, when $a_{k,n}=1$, device $k$ will harvest energy from time slot $t_0$ to $t_{n-1}$. 
Once task $k$ arrives at the MEC server, i.e., at time slot slot $t_{n+1}$,  the server will process this computational task. 
 
 The task computation process of the server and a device jointly completing a computational task $k$ consists of three stages: 1) energy harvesting, 2) task offloading, and 3) remote computing. 
 Next, we first introduce the process of the energy harvesting, task offloading, and remote computing stages. Then, the problem formulation is given.

 \subsection{Energy Harvesting Model}
  The path loss model is given by $\bar{h}_k=  A\left(\frac{c}{4\pi f_c d_k}\right)^{\ell}$, where $A$  represents antenna gain, $c$ denotes the speed of light,  $f_c$ is the carrier frequency, $\ell$ denotes the path-loss factor, and $d_k$ represents the distance between device $k$ and the server \cite{9449944}. The instant channel gain between device $k$ and server  denoted by $h_k$, follows an  i.i.d. Rician distribution with line-of-sight (LoS)	link gain equal to $\gamma\bar{h}_k$, where $\gamma$ is Rician factor. 
If device $k$ offloads its task at time slot $t_n$ (i.e., $a_{k,n}=1$), the harvested energy of device $k$ is $E_k^H=\sum_{i=0}^{n-1}\Delta{t_i}h_k\eta P_0$, where $\eta$ is the energy harvesting efficiency of each device, which is assumed to be equal for all devices \cite{8960510}. $P_0$ denotes the transmit power of the server. 
Since each device has only a single time slot for task offloading (i.e., there exists only one $n\in\mathcal{N}$ such that $a_{k,n}=1$ for a certain device $k$), the energy harvested by device $k$ can be reformulated by 
$E_{k}^H=\sum_{n=1}^K\sum_{i=0}^{n-1}a_{k,n}\Delta{t_i}h_k\eta P_0$ $ (\forall k\in\mathcal{K})$.
	
	\subsection{Tasks Offloading Model}
	Based on the monomial offloading power model \cite{8468240,9140412}, the transmit power of device $k$ at its offloading time slot $t_n$ is
	\begin{align}
		p_{k,n}=\frac{\lambda(r_{k,n})^3}{h_k}=\frac{\lambda(A_k)^3}{h_k(\Delta t_n)^3}, \forall k\in\mathcal{K},\forall n\in\{1,\cdots,N\}, 
	\end{align}
	where $r_{k,n}=A_k/\Delta t_n$ is the transmission rate, 
 $\lambda>0$ is the energy coefficient related to the bandwidth and the noise power, and the order $3$ is the monomial order associated with coding scheme.
	Since the transmit power of devices comes from harvested energy, we have 
	$	\sum_{n=1}^{N}a_{k,n}\Delta t_np_{k,n}\leq E_k^H$ $(\forall k\in\mathcal{K})$.



	\subsection{Computing Model}
 		\begin{figure}[t]
		\centering
		\includegraphics[width=0.45\textwidth]{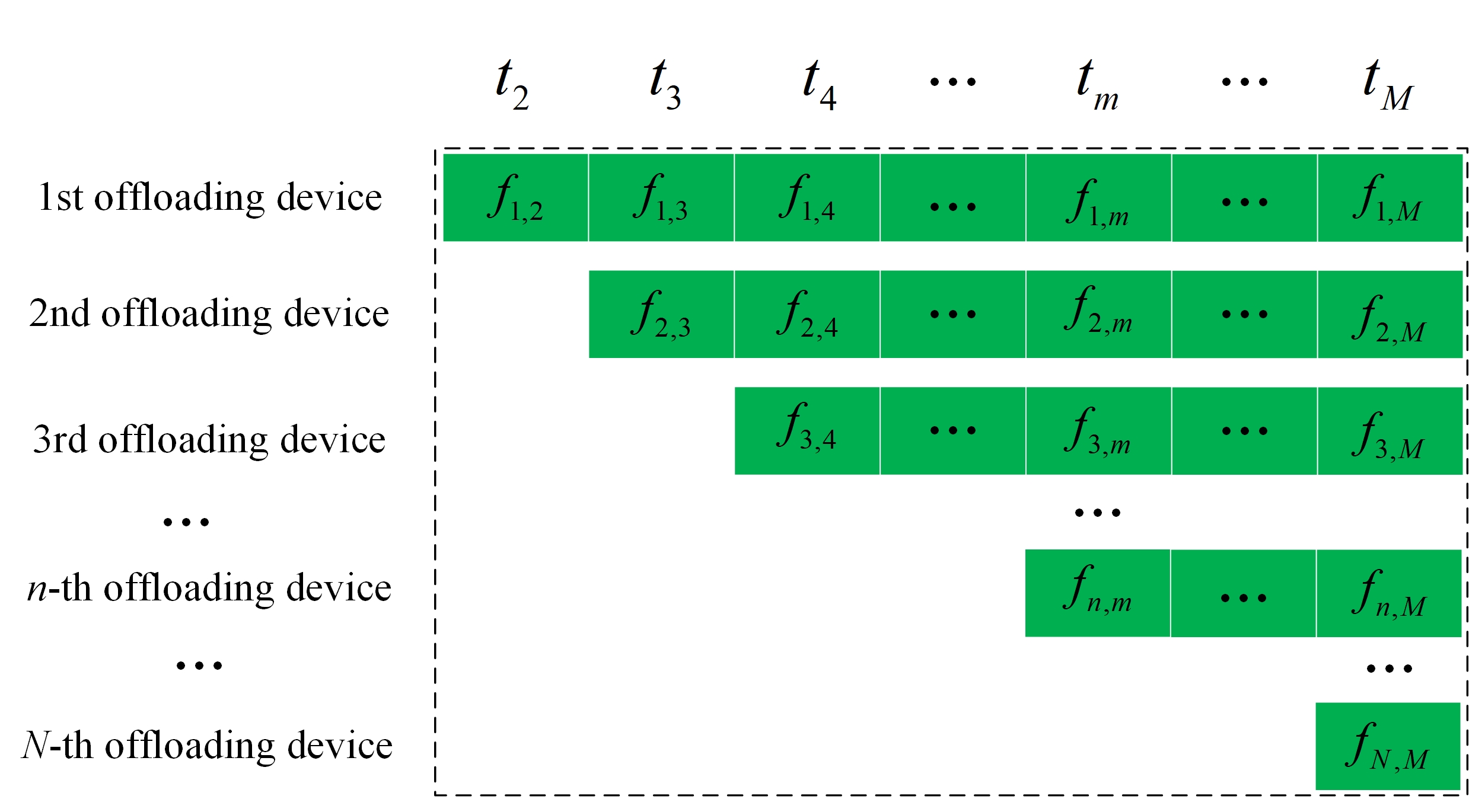}
		\caption{Asynchronous computation resource allocation.} \label{fig_fnm}
	\end{figure}
The MEC server is equipped with multiple CPUs such that the computational tasks offloaded from different devices can be executed in parallel \cite{8964328}. Let $I_k$ be the computation intensity of task $k$ in terms of CPU cycles per bit. As shown in Fig.~\ref{fig_fnm}, to sufficiently utilize asynchronous computing, the computation resource of the server will be  reallocated to the computational tasks offloaded from devices at each time slot from $t_2$ to $t_{N+1}$. 
Intuitively, the first uploading task can occupy  the whole computation capacity of the server before the second offloading task arrives,  while all the tasks compete for computation resource at time slot $t_{N+1}$.
At an arbitrary time slot $m \in\{2,\cdots,N+1\}$, $f_{n,m}$, $\forall n\in\{1,\cdots,m-1\}$ is set to be the computation resource  allocated to the task that arrives at the server at time slot $t_n$. 
Given these definitions, we have 
$\sum_{n=1}^{m-1} f_{n,m}\leq F_{\max}$ $(\forall m = 2, \cdots, N+1)$, 
where $F_{\max}$ represents the maximum computation capacity of the MEC server. To complete the task computation for each device $k$, we have
$\sum_{n=1}^N \sum_{m=n+1}^{N+1} a_{k,n}f_{n,m}\Delta t_m \geq F_k$ $(\forall k\in\mathcal{K})$, 
where $F_k=A_kI_k$ represents the computation cycles of device $k$.  
Besides, the energy consumption  of the MEC server for all tasks computation  can be formulated by
$E_{MEC}=\sum_{n=1}^N \sum_{m=n+1}^{N+1}  \kappa f_{n,m}^3\Delta t_m$, 
where $\kappa$ denotes the energy coefficient of the MEC server.

\subsection{Problem Formulation}
Our goal is to minimize the MEC server's energy consumption of completing the   tasks offloaded by all devices, which is formulated as an optimization problem as
\begin{subequations} \label{P1}
	\begin{align}
		\min_{\Delta\pmb{t},\pmb{A},\pmb{f}}\  &\sum_{n=1}^N \sum_{m=n+1}^{N+1}  \kappa f_{n,m}^3\Delta t_m,\tag{\ref{P1}}\\
		\text{s.t.}\   
		&\sum_{n=1}^{N}a_{k,n} \frac{\lambda(A_{k})^3}{h_k(\Delta t_n)^2}\leq \sum_{n=1}^N\sum_{i=0}^{n-1}a_{k,n}\Delta t_ih_k\eta P_0,  \forall k\in\mathcal{K},\\
		&\sum_{n=1}^{m-1} f_{n,m}\leq F_{\max}, \quad \forall m = 2, \cdots, N+1,\\
		&\sum_{n=1}^N \sum_{m=n+1}^{N+1} a_{k,n}f_{n,m}\Delta t_m \geq A_kI_k, \quad \forall k\in\mathcal{K},\\
		&\sum_{i=0}^{N+1}\Delta t_i\leq T,\\
		&\sum_{k=1}^K a_{k,n}=1,\quad \forall n\in\{1,\cdots,N\},\\
		&\sum_{n=1}^N a_{k,n}=1,\quad \forall k\in\mathcal{K},\\
		&a_{k,n}\in\{0,1\}, \quad \forall k\in\mathcal{K},\forall n\in\{1,\cdots,N\},
	\end{align}
\end{subequations}
where $\Delta\pmb{ t}=[\Delta t_0,\cdots,\Delta t_{N+1}]^T$,  $\pmb{f}=(f_{n,m})_{\forall n\in\{1,\cdots,N\}, m \in\{n+1,\cdots,N+1\}}$,  and  $\pmb{A}=(a_{k,n})_{K\times N}$.
In (\ref{P1}), (\ref{P1}a) is energy consumption causality constraint; (\ref{P1}b) represents a computational resource allocation constraint; 
(\ref{P1}c) ensures the completion of task computing; 
(\ref{P1}d) implies that the execution time of all devices should be less than $T$;  (\ref{P1}e)-(\ref{P1}g) are user scheduling  constraints. 
 Since the discrete user scheduling variables $a_{k,n}$ and continuous resource allocation variables $\Delta t_i$, $f_{n,m}$ are highly coupled,  problem \eqref{P1} is a standard MINLP problem which is difficult to solve.  To handle this issue, we first analyze the optimal computation resource allocation with given user scheduling and time allocation in Section III,  based on which an efficient  low-complexity computation frequency optimization algorithm is proposed. Finally, in Section IV, we propose a GBD-based algorithm to jointly optimize user scheduling and resource allocation so as to solve problem \eqref{P1}. \footnote{ For multi-server edge computing systems, new indicator variables can be introduced to denote the association between tasks and servers.  
		Then the energy minimization problem can be  formulated as a MINLP problem containing two kinds of binary optimization variables for task-server association and  scheduling sequence, respectively. Despite being more complex, the problem can be solved efficiently using conventional MINLP  methods such as convex relaxation and branch-and-bound, or latest approach using machine learning (see e.g.,  \cite{9449944}). It is worth noting that with given task-server association, the proposed algorithm in this work is still applicable to scheduling and resource allocation optimization for each server. The detailed transmission protocol and  algorithm procedure are left for future works. } 

\section{Analysis and Algorithm of the Optimal Computation Resource Allocation}
In this section, we first analyze the  properties of the optimal computation resource allocation,  and then a low-complexity computation resource allocation algorithm is accordingly proposed. For ease of notation, we use $F_n$ $(\forall n\in\{1,\cdots,N\})$ to represent the computation cycles to complete the task that arrives at the server with order $n$. With given time slot allocation vector $\Delta \pmb{t}$ and user scheduling matrix $\pmb{A}$, problem \eqref{P1} is simplified as follows: 
\begin{subequations} \label{theorem4_eq1}
	\begin{align}
		\min_{\{f_{n,m}\}}\  &\sum_{n=1}^N \sum_{m=n+1}^{N+1}  \kappa f_{n,m}^3\Delta t_m,\tag{\ref{theorem4_eq1}}\\
		s.t.\ \   
		&\sum_{n=1}^{m-1} f_{n,m}\leq F_{\max}, \quad \forall m = 2, \cdots, N+1,\\
		&\sum_{m=n+1}^{N+1} f_{n,m}\Delta t_m \geq F_{n}, \quad \forall n\in\{1,\cdots,N\},\\
		&f_{n,m}\geq 0,\forall n\in\{1,\cdots,N\},\forall m = n+1, \cdots, N+1.
	\end{align}
\end{subequations}
Before solving problem \eqref{theorem4_eq1}, we provide the feasibility  condition as follows. 

\begin{proposition} \label{proposition_1}
	\emph{Problem \eqref{theorem4_eq1} is feasible if and only if $F_{\max}\geq \max_{n\in\{1,\cdots,N\}} \frac{\sum_{i=n}^NF_i}{\sum_{i=n+1}^{N+1}\Delta t_{i+1}}$.}
\end{proposition}
\begin{proof}
	Please refer to Appendix \ref{proof_proposition_1}.
\end{proof}

Denote $\{\alpha_m\}$, $\{\beta_n\}$, and $\{\gamma_{n,m}\}$ 
as the non-negative Lagrangian multipliers associated with
the maximum frequency constraints (\ref{theorem4_eq1}a), task computation completion
constraints  (\ref{theorem4_eq1}b) and non-negative frequency constraints  (\ref{theorem4_eq1}c),
respectively. The optimal computation resource allocation is given by the following proposition. 

\begin{proposition}
  \label{theorem_10}
	\emph{
	Given the optimal $\{\alpha^*_{m}\}$, $\{\beta^*_n\}$, the optimal solution of problem \eqref{theorem4_eq1} is given by
	\begin{align} \label{equation4}
		f^*_{n,m}=\sqrt{\left[\frac{\beta^*_n}{3\kappa}-\frac{\alpha^*_m}{3\kappa \Delta t_m}\right]^+},\forall n\in\mathcal{K},\forall m\in\{2,\cdots,K+1\}. 
	\end{align}}
\end{proposition}
\begin{proof}
Since \eqref{equation4} can be effectively  obtained by solving Karush-Kuhn-Tucker (KKT)  and Slater conditions, the proofs is omitted here. 
\end{proof}

According to Propostion~\ref{theorem_10}, we can use the sub-gradient method to obtain the optimal $\{\alpha^*_m\} $ and $\{\beta^*_n\}$ so as to acquire the optimal computation resource allocation.   To further reduce the computational complexity and provide some design insights, the properties of the optimal solution of problem \eqref{theorem4_eq1} are summarized in the following theorem.

\begin{theorem}
    \emph{
    Denote  $\digamma(i)=\sum_{n=1}^{i-2}\frac{F_n}{\sum_{m=n+1}^{K+1}\Delta t_m} + \frac{\sum_{n=i-1}^K F_n}{\sum_{m=i}^{K+1}\Delta t_m}$ $(2\leq i\leq K+1)$. 
 The optimal computation resource has the following properties:
    \begin{itemize}\label{theorem_5}  
        \item[1)] The optimal solution of problem \eqref{theorem4_eq1} satisfies   $f_{n,n+1}^*=\cdots= f_{n,i}^*>\cdots>f^*_{n,j}=\cdots= f_{n,K+1}^*=0,$  $(n+1\leq i < j\leq K+1)$, where $t_i$ is referred as ``transition point". 
        \item[2)] The optimal $\{\alpha^*_m\} $ satisfies $0=\frac{\alpha^*_2}{\Delta t_2}=\cdots=\frac{\alpha^*_i}{\Delta t_i}<\frac{\alpha^*_{i+1}}{\Delta t_{i+1}}\cdots<\frac{\alpha_{K+1}^*}{\Delta t_{K+1}}$.
        \item[3)] The transition point is $t_i$ $(3\leq i\leq K+1)$ if and only if 
$\digamma(i-1)\leq F_{\max} < \digamma(i)$. 
    \end{itemize}
    }
\end{theorem}
\begin{proof}
	The proofs of 1), 2) and 3) are provided in Appendix \ref{proof_theorem_5}, \ref{proof_theorem1_2)}, \ref{proof_theorem_8},  respectively.
\end{proof}


According to  property  1) in  Theorem~\ref{theorem_5}, the optimal frequency allocation
scheme for a certain offloading task always follows a specific pattern: 
the frequency allocated to each device remains constant initially, then gradually decreases and eventually reaches 
zero.  
This property  motivates us to deduce the condition  
$f^*_{n,i}>f^*_{n,i+1}$. 
The property 2) in Theorem~\ref{theorem_5} implies that the computation resource of the server is redundant at time slots from  $t_2$ to $t_{i}$, while the maximum computation resource $F_{\max}$ is utilized  at   slots from $t_{i+1}$ to $t_{K+1}$.   

According to property  1) in  Theorem~\ref{theorem_5}, unless $f^*_{n,n+1}=\cdots=f^*_{n,K+1}$ $(\forall n\in\mathcal{K})$, there always exists a special time slot $t_{\varkappa}$ we called ``\textit{transition point}'' such that $f^*_{n,n+1}=\cdots=f^*_{n,\varkappa-1}>f^*_{n,\varkappa}\geq\cdots\geq f^*_{n,K+1}$ $(3\leq \varkappa \leq K+1)$. The transition point indicates the number of time slots that the computation resource remains the same.  The computation resource decreases for all tasks at the transition point. 
The method to find out the transition point when it exists is given by property 3) in Theorem~\ref{theorem_5}. 

Property 3) in Theorem~\ref{theorem_5} also shows that the transition point is impacted by the computation ability of the  server. 
We can directly  determine the transition point $t_\varkappa$   utilizing property 3) in Theorem~\ref{theorem_5} without the need of solving problem \eqref{theorem4_eq1}. After determining the transition point, we have $\alpha^*_m=0$ $(2\leq m\leq\varkappa-1)$ according to property 2) in Theorem~\ref{theorem_5}. 


\begin{figure}[t]
	\centering
	\includegraphics[width=0.47\textwidth]{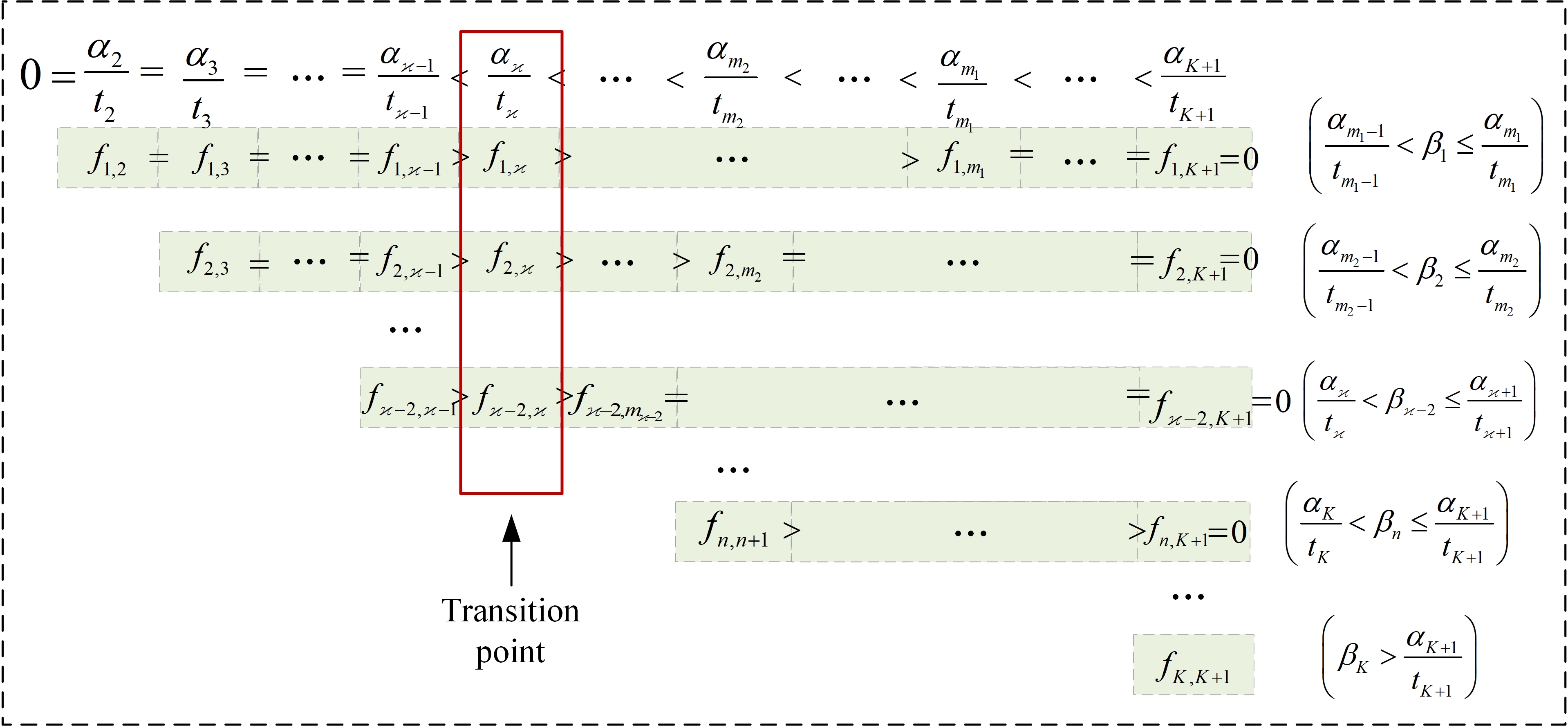}
	\caption{Properties of the optimal computation resource allocation.} \label{theorem_fig}
\end{figure}

Fig. \ref{theorem_fig} depicts an illustration of  properties in the optimal computation resource allocation. As can be seen, before the transition point $t_{\varkappa}$, the optimal $\alpha^*_m=0$ and $f^*_{n,m}$ keeps unchanged  as $m=2,\cdots,\varkappa-1$. 
Based on  Theorem~\ref{theorem_5},  a low-complexity algorithm is proposed in Algorithm \ref{alg1}. First, we check the feasibility of problem \eqref{theorem4_eq1} according to Proposition \ref{proposition_1}. Then, we determine the transition point $t_{\varkappa}$ based on property 3) in Theorem~\ref{theorem_5}. If there is no transition points, which means that the computation resource of server is abundant, we can directly obtain the optimal solution $f^*_{n,n+1}=\cdots=f^*_{n,K+1}=\frac{F_n}{\sum_{m=n+1}^{K+1}\Delta t_m}$ for $n=1,\cdots,K$; otherwise, we obtain the transition point $t_{\varkappa}$ and have  $\alpha^*_m=0$ for $m=2,\cdots,\varkappa-1$. 
Hence, we only need to find out the optimal $\alpha^*_m$ $(m=\varkappa,\cdots,K+1)$ and $\beta^*_n$ $(n=1,\cdots,K)$. Note that with given $\beta^*_n$ $(n=1,\cdots,K)$, we can obtain the optimal $\alpha^*_m$ by solving the following $(K-\varkappa+2)$ equalities 
\begin{align} \label{equality}
	G(m,\alpha_m)\triangleq\sum_{n=1}^{m-1}&\sqrt{\left[\frac{\beta_n}{3\kappa}-\frac{\alpha_m}{3\kappa \Delta t_m}\right]^+}=F_{\max},  \nonumber\\
 &\quad\quad\quad(m=\varkappa,\cdots,K+1),
\end{align}
 since the maximum frequency is utilized at time slots from $t_{\varkappa}$ to $t_{K+1}$. 
 Since $G(m,\alpha_m)$  decreases with respect to $\alpha_m$, the bisection method is adopted. 
 It should be noticed that $G(m,\alpha_m)$ achieves the maximum value of  $\sum_{n=1}^{m-1}\sqrt{\frac{\beta_n}{3\kappa}}$ when $\alpha_m=0$
and the minimum value of $0$ when $\frac{\alpha_m}{\Delta t_m}\geq\max_{n=1,\cdots,m-1} \beta_n$. Therefore, the upper bound of $\frac{\alpha_m}{\Delta t_m}$ is set as $\frac{\alpha_m^{ub}}{\Delta t_m}=\max_{n=1,\cdots,m-1} \beta_n$.  For the lower bound, 
we set $\frac{\alpha_m^{lb}}{\Delta t_m}=\frac{\alpha_{m-1}}{\Delta t_{m-1}}$ according to property 2) in  Theorem~\ref{theorem_5}. 
After obtaining $\alpha_m$ for $m=\varkappa,\cdots,K+1$, $\beta_n$ is updated by a  sub-gradient method \cite{cot}, where $\phi_n$ is the dynamically chosen step-size.  Through repeating Steps 5 to 13 until the objective of \eqref{theorem4_eq1} converges, we can obtain the optimal $\alpha^*_m$ for $m=2,\cdots,K+1$ and $\beta^*_n$ for $n=1,\cdots,K$.

The complexity of Algorithm \ref{alg1} is $\mathcal{O}\left(\frac{(K+2-\varkappa)}{\sqrt{\epsilon_1}}\log_2(\frac{1}{\epsilon_0})\right)$, where $\epsilon_0$ denotes the accuracy of the bisection method and $\epsilon_1$ is the accuracy of the objective of problem  \eqref{theorem4_eq1}. Compared with the complexity of  $\mathcal{O}\left(\left(K^2+K\right)^{3.5}\right)$ by the interior point method, the complexity of the proposed algorithm is significantly reduced.  Moreover, when $\varkappa$ is large, the  complexity can be further reduced since more numbers of $\alpha^*_m$ are zeros.

\begin{algorithm}[t]
	\begin{small}
	\caption{Optimal Computation Frequency Allocation Algorithm}
	\label{alg1}
	
	 If $F_{\max}\geq  \max_{n\in\{1,\cdots,K\}} \frac{\sum_{i=n}^KF_i}{\sum_{i=n+1}^{K+1}\Delta t_{i+1}}$, go to  Step 2; otherwise, problem \eqref{theorem4_eq1} is infeasible. 
	
	  According to Theorem \ref{theorem_5}, if there is no transition point, the optimal solution is given by $f^*_{n,n+1}=\cdots=f^*_{n,K+1}=\frac{F_n}{\sum_{m=n+1}^{K+1}\Delta t_m}$ for $n=1,\cdots,K$; otherwise, obtain transition point $t_{\varkappa}$ and let  $\alpha^*_m=0$ for $m=2,\cdots,\varkappa-1$. 
	
	  Initialize $\beta_n=\left(\frac{F_n}{\sum_{m=n+1}^{K+1}\Delta t_m}\right)^2$ for $n=1,\cdots,K$ and required precision $\epsilon_0$.  
	
	  \Repeat{\textnormal{the objective of \eqref{theorem4_eq1} converges}}{ 
	  	
	  \For{$m=\varkappa,\cdots,K+1$}{
	    Let  $\frac{\alpha_m^{lb}}{\Delta t_m}=\frac{\alpha_{m-1}}{\Delta t_{m-1}}$ and $\frac{\alpha_m^{ub}}{\Delta t_m}=\max_{n=1,\cdots,m-1}\beta_n$.  
	 
	  \While{$\frac{\alpha_m^{ub}}{\Delta t_m}-\frac{\alpha_m^{lb}}{\Delta t_m}>\epsilon_0$}{ 
	    Set $\frac{\alpha_m}{\Delta t_m}\leftarrow(\frac{\alpha_m^{lb}}{\Delta t_m}+\frac{\alpha_m^{ub}}{\Delta t_m})/2$. 
	   
	    Calculate $f_{n,m}=\sqrt{\left[\frac{\beta_n}{3\kappa}-\frac{\alpha_m}{3\kappa \Delta t_m}\right]^+}$. 
	   
	    If $\sum_{n=1}^{m-1}f_{n,m}>F_{\max}$, let  $\frac{\alpha_m^{lb}}{t_m}\leftarrow\frac{\alpha_m}{\Delta t_m}$; otherwise, let  $\frac{\alpha_m^{ub}}{t_m}\leftarrow\frac{\alpha_m}{\Delta t_m}$.
		}
		}

	  Update $\beta_n\leftarrow\left[\beta_n+\phi_n\left(F_n-\sum_{m=n+1}^{K+1} f_{n,m}\Delta t_m\right)\right]^+$ for $n=1,\cdots,K$.
	 
 	} 
	Output  the optimal $\{f^*_{n,m}\}$. 
\end{small}  
\end{algorithm}

\section{Joint User Scheduling and Resource Allocation Algorithm}
 In this section, we employ the  GBD method to solve  problem \eqref{P1}. The core idea of GBD method is decomposing the original MINLP problem into a primal problem related to continuous variables and a master problem associated with integer variables, which are  iteratively solved\footnote{Interested readers may refer to \cite{9140412,9399826,8910634,7372448} for details. }. Specifically, for problem \eqref{P1}, the primal problem is a joint communication and computation resource optimization problem with fixed user scheduling. The master problem optimizes user scheduling by utilizing the optimal solutions and dual variables of the primal problem. Next, we describe the detailed procedures. 

\subsection{Primal Problem}
With given user scheduling $\pmb{A}$, problem \eqref{P1} is reduced to the following optimization problem:
	\begin{subequations} \label{P2}
	\begin{align}
		\min_{\Delta \pmb{t},\pmb{f}}\quad &\sum_{n=1}^K \sum_{m=n+1}^{K+1}  \kappa f_{n,m}^3\Delta t_m,\tag{\ref{P2}}\\
		s.t.\quad  
		& \frac{\lambda(A_{\pi_n})^3}{h_{\pi_n}(\Delta t_n)^2}\leq \sum_{i=0}^{n-1}\Delta t_ih_{\pi_n}\eta P_0, \quad \forall n\in\mathcal{K},\\
		& \sum_{m=n+1}^{K+1} f_{n,m}\Delta t_m \geq A_{\pi_n}I_{\pi_n}, \quad \forall n\in\mathcal{K},\\
		&\text{(\ref{P1}b)}, \text{(\ref{P1}d)}, 
	\end{align}
\end{subequations}
  where $\pi_n$ denotes the index of the $n$-th offloading device, i.e., we have $\pi_n=k$ if $a_{k,n}=1$. Since the user scheduling scheme $\pmb{A}$ is known, the value of $\pi_n,(\forall n \in\mathcal{K})$ can be deduced and substituted into problem \eqref{P1}. Since problem \eqref{P2} is non-convex due to the constraints (\ref{P2}a),  (\ref{P2}b) and the objective, we introduce  $x_{n,m}=f_{n,m}\Delta t_m (\forall n=1,\cdots,K, \forall m=n+1,\cdots,K+1)$ to represent computation amounts of the $n$-th offloading task at time slot $m$. Hence, problem \eqref{P2} is equivalent to 
\begin{subequations} \label{P2.1}
	\begin{align}
		\min_{\Delta \pmb{t},\pmb{x}}\quad &\sum_{n=1}^K \sum_{m=n+1}^{K+1}  \kappa \frac{(x_{n,m})^3}{(\Delta t_m)^2},\tag{\ref{P2.1}}\\
		s.t.\quad  
		&\sum_{n=1}^{m-1} x_{n,m}\leq F_{\max}\Delta t_m,  \forall m = 2, \cdots, K+1,\\
		&\sum_{m=n+1}^{K+1} x_{n,m} \geq A_{\pi_n}I_{\pi_n}, \quad \forall n\in\mathcal{K},\\
		&\text{(\ref{P1}d)}, \text{(\ref{P2}a)}, 
	\end{align}
\end{subequations}
where  $\pmb{x}=(x_{n,m})_{\forall n\in\mathcal{K}, m \in\{n+1,\cdots,K+1\}}$ is the collections of $x_{n,m}$.  It can be proved that problem \eqref{P2.1} is convex utilizing the tricks of  perspective function \cite{boyd2004convex}. To further provide useful insights and reduce computation complexity, we utilize the block coordinate decent (BCD) method to iteratively optimize time allocation  and computation resource. Since the low-complexity computation resource allocation algorithm with given time allocation has been provided in Algorithm~\ref{alg1}, next we propose time allocation algorithm with fixed computation resource  allocation.  

The Lagrangian function of problem \eqref{P2.1} with respect to $\Delta\pmb{t}$ is given by

\begin{small}
\begin{align}
	\mathcal{L}=&\sum_{n=1}^K \sum_{m=n+1}^{K+1}\!\!\kappa \frac{(x_{n,m})^3}{(\Delta t_m)^2} \!+\!\sum_{n=1}^K\rho_n\!\left(\!\frac{\lambda(A_{\pi_n})^3}{h_{\pi_n}(\Delta t_{n})^{2}}\!- \!\sum_{i=0}^{n-1}\Delta t_ih_{\pi_n}\eta P_0\right)\nonumber\\
	&+\xi\left(\sum_{i=0}^{K+1}\Delta t_i- T\right) + \sum_{m=2}^{K+1}\omega_m\left(\sum_{n=1}^{m-1} x_{n,m}- F_{\max}\Delta t_m\right), 
\end{align}
\end{small}\\
where $\rho_n$, $\omega_m$ and $\xi$  are dual variables related to constraints (\ref{P2}a), (\ref{P2.1}a) and (\ref{P1}d), respectively. Taking the derivative with respect to $\Delta \pmb{t}$, we have

\begin{small}
\begin{align}
	\frac{\partial\mathcal{L}}{\partial\Delta t_0^*}&=-\sum_{n=1}^K\rho^*_nh_{\pi_n}\eta P_0 + \xi^*=0 \label{Lag_0},\\
	\frac{\partial\mathcal{L}}{\partial\Delta t^*_1}&=-\rho^*_1\frac{2\lambda(A_{\pi_1})^3}{h_{\pi_1}(\Delta t^*_{1})^{3}} - \sum_{n=2}^K\rho^*_n h_{\pi_n} \eta P_0 +\xi^*=0 \label{Lag_1},\\
	\frac{\partial\mathcal{L}}{\partial\Delta t^*_i}&=-\sum_{n=1}^{i-1}2\kappa\frac{(x_{n,i})^3}{(\Delta t^*_i)^3}-\rho^*_i\frac{2\lambda(A_{\pi_i})^3}{h_{\pi_i}(\Delta t^*_{i})^{3}} - \sum_{n=i+1}^K\rho^*_n h_{\pi_n} \eta P_0 \nonumber\\
 &\quad\ +\xi^* -\omega^*_iF_{\max}=0, \quad (2\leq i \leq K-1)\label{Lag_i} ,\\
	\frac{\partial\mathcal{L}}{\partial\Delta t^*_{K}}&\!=\!-\!\!\sum_{n=1}^{K-1}\!2\kappa\frac{(x_{n,K})^3}{(\Delta t^*_K)^3}\!-\!\rho^*_K\frac{2\lambda(A_{\pi_K})^3}{h_{\pi_K}(\Delta t^*_{K})^{3}} \!+\!\xi^* \!-\!\omega^*_KF_{\max}\!=\!0\label{Lag_K}, \\
	\frac{\partial\mathcal{L}}{\partial\Delta t^*_{K+1}}&=-\sum_{n=1}^{K}2\kappa\frac{(x_{n,K+1})^3}{(\Delta t^*_{K+1})^3}+\xi^* -\omega^*_{K+1}F_{\max}=0 \label{Lag_K+1}. 
\end{align}
\end{small}\\
Through solving above equations, the optimal $\Delta t^*_i$ $\left(\forall i\right)$ is obtained in the following proposition. 

\begin{proposition}\label{proposition_t}
    \emph{The optimal $\Delta\pmb{t}^*$ is given by 
    \begin{align}
        &\Delta t_0^*=T-\sum_{i=1}^{K+1}\Delta t_i^*,\\
        &\Delta t^*_{i}=\sqrt[3]{\frac{2\kappa h_{\pi_i}\sum_{n=1}^{i-1}(x_{n,i})^3+2\rho^*_{i}\lambda (A_{\pi_i})^3}{\xi^*h_{\pi_i} - \sum_{n=i+1}^K\rho^*_n (h_{\pi_n})^2 \eta P_0  -\omega^*_ih_{\pi_i}F_{\max}}},\nonumber \\
        &\quad\quad\quad\quad\quad\quad\quad\quad\quad(\forall i \in \{2,\cdots,K-1\}),\\
        &\Delta t^*_K=\sqrt[3]{\frac{2\kappa h_{\pi_K}\sum_{n=1}^{K-1}(x_{n,K})^3+2\rho_K^*\lambda(A_{\pi_K})^3}{\xi^*h_{\pi_K}-\omega^*_Kh_{\pi_K}F_{\max}}},\\
        &\Delta t_{K+1}^*=\sqrt[3]{\frac{2\kappa \sum_{n=1}^K(x_{n,K+1})^3}{\xi^*-\omega^*_{K+1}F_{\max}}}, 
    \end{align}
    and $\Delta t_1^*$ is the null point of $\Psi(x)$, where $\Psi(x)=\left(T-\sum_{i=1}^{K+1}\Delta t_i\right)x^2(h_{\pi_1})^2\eta P_0-\lambda(A_{\pi_1})^3$,   $x\in\Big(0,\frac{2}{3}\left(T-\sum_{i=2}^{K+1}\Delta t_i\right)\Big]$. 
    }
\end{proposition}
\begin{proof}
    The proof of Proposition~\ref{proposition_t} is provided in Appendix~\ref{proof_proposition_t}. 
\end{proof}

Through iteratively optimizing time allocation and computation resource  allocation, we can obtain the optimal solution of primal problem \eqref{P2.1}. 
However, if problem \eqref{P2.1} is infeasible, we formulate the corresponding $\ell_1$-minimization problem as follows: 
\begin{subequations} \label{P2_feasible}
	\begin{align}
		\min_{\Delta \pmb{t},\pmb{x}, 
  \pmb{\zeta}>0,\pmb{\iota}>0}&\sum_{k=1}^K \left(\zeta_k+\iota_k\right),\tag{\ref{P2_feasible}}\\
		s.t.\quad\    
		&\sum_{n=1}^{K}a_{k,n} \frac{\lambda(A_{k})^3}{h_k(\Delta t_n)^2}\!\leq\! \zeta_k\! +\! \sum_{n=1}^K\sum_{i=0}^{n-1}a_{k,n}\Delta t_ih_k\eta P_0, \nonumber\\
  &\quad\quad\quad\quad\quad\quad\quad\quad\quad\quad\quad\quad \forall k\in\mathcal{K},\\
		&\iota_k\!+\!\sum_{n=1}^K \sum_{m=n+1}^{K+1}\! a_{k,n}x_{n,m} \!\geq\! A_kI_k, \forall k\in\mathcal{K},\\
		&\text{(\ref{P1}d)}, \text{(\ref{P2.1}a)}.  
	\end{align}
\end{subequations}
Since problem \eqref{P2_feasible} is convex and always feasible, we can use the interior point method to obtain the optimal solution and corresponding dual variables. 

{ 
Furthermore, we can observe that the solution of primal problem always provides a performance upper bound for problem \eqref{P1} since user scheduling is fixed. Then the upper bound is updated as $UB^{(j)}\leftarrow\min\{UB^{(j-1)},f^{(j)}\}$, where $f^{(j)}$ denotes the objective value of primal problem \eqref{P2}. As can be seen, the upper bound is always non-increasing as iteration proceeds.  Subsequently, we construct master problem using the solutions and dual variables of primal problem \eqref{P2.1} and feasibility problem \eqref{P2_feasible}.

\subsection{Master Problem}
At each iteration, optimality cut or feasibility cut are added to master problem depending on whether the primal problem is feasible. Denote $\mathcal{J}_1$ and $\mathcal{J}_2$ as the set of iteration indexes indicating the primal problem is feasible and infeasible, respectively. 
Specifically, 
the optimality cut for each  $j\in\mathcal{J}_1$ of feasible iterations is defined as 
 
 \begin{small}
\begin{align} \label{opt_cut}
	&\theta(\pmb{A},\pmb{\rho}^{(j)},\pmb{\beta}^{(j)})=\sum_{n=1}^K \sum_{m=n+1}^{K+1}  \kappa \frac{\left(x^{(j)}_{n,m}\right)^3}{\left(\Delta t^{(j)}_m\right)^2} \nonumber\\
 &+\sum_{k=1}^K\rho^{(j)}_k\left(\sum_{n=1}^{K}a_{k,n} \frac{\lambda(A_{k})^3}{h_k\left(\Delta t^{(j)}_n\right)^2}- \sum_{n=1}^K\sum_{i=0}^{n-1}a_{k,n}\Delta t^{(j)}_ih_k\eta P_0\right)\nonumber\\
	&+\sum_{k=1}^K \beta_k^{(j)}\left( A_kI_k-\sum_{n=1}^K \sum_{m=n+1}^{K+1} a_{k,n}x^{(j)}_{n,m}\right),
\end{align} 
\end{small}\\
where $\rho_k^{(j)}$ and $\beta_k^{(j)}$ represent the dual variables related to primal problem at the $j$-th iteration, $x_{n,m}^{(j)}$ and  $\Delta t^{(j)}_m$ denote the solution of primal problem at the $j$-th iteration. The terms irrelavant to $\pmb{A}$ are omitted based  on complementary slackness theorem  \cite{cot}. 
Similarly, the feasibility cut for each  $j\in\mathcal{J}_2$ of infeasible iterations is defined as 

\begin{small}
\begin{align} \label{feas_cut}
	&\hat{\theta}(\pmb{A},\hat{\pmb{\rho}}^{(j)},\hat{\pmb{\beta}}^{(j)})=\sum_{k=1}^K \hat{\beta}_k^{(j)}\left( A_kI_k-\sum_{n=1}^K \sum_{m=n+1}^{K+1} a_{k,n}\hat x^{(j)}_{n,m}\right)\nonumber\\
 &+\sum_{k=1}^K\hat{\rho}^{(j)}_k\left(\sum_{n=1}^{K}a_{k,n} \frac{\lambda(A_{k})^3}{h_k\left(\Delta \hat t^{(j)}_n\right)^2}- \sum_{n=1}^K\sum_{i=0}^{n-1}a_{k,n}\Delta \hat t^{(j)}_ih_k\eta P_0\right),
\end{align} \end{small}\\
where  $\hat\rho_k^{(j)}$ and $\hat\beta_k^{(j)}$ represent the dual variables related to feasibility problem at the $j$-th iteration, $\hat x_{n,m}^{(j)}$ and  $\Delta \hat t^{(j)}_m$ denote the solution of feasibility problem at the $j$-th iteration. 
Therefore, master problem is formulated as 
	\begin{subequations} \label{master}
	\begin{align}
		\min_{\pmb{A},\psi}\quad &\psi,\tag{\ref{master}}\\
		s.t.\quad  
		&\theta(\pmb{A},\pmb{\rho}^{(j)},\pmb{\beta}^{(j)})\leq\psi,\quad\forall j\in\mathcal{J}_1,\\
		&\hat{\theta}(\pmb{A},\hat{\pmb{\rho}}^{(j)},\hat{\pmb{\beta}}^{(j)})\leq0,\quad\forall j\in\mathcal{J}_2,\\
		&\text{(\ref{P1}e)}- \text{(\ref{P2}g)}. 
	\end{align}
\end{subequations}
In particular, (\ref{master}a) and (\ref{master}b)  denote the set of  hyperplanes spanned by the
optimality cut and feasibility cut from the first to the $j$-th iteration, respectively. The two different types of cuts are exploited
to reduce the search region for the global optimal solution \cite{7106496}.  
Master problem \eqref{master} is a standard mixed-integer linear programming (MILP) problem, which can be solved by numerical solvers such as Gurobi\cite{Gurobi} and Mosek\cite{Mosek}.  
Since master problem is the relaxing problem of MINLP problem \eqref{P1},  solving master problem provides a performance lower bound for problem \eqref{P1}. 
The lower bound is given by $LB^{(j)}\leftarrow\psi$. Since at each iteration, an additional cut (optimality cut or feasibility cut) is added to master problem which narrows the feasible zone, the lower bound is always non-decreasing. 
 As a consequence, the performance upper
bound obtained by primal problem and the performance
lower bound obtained by the master problem  are non-increasing and non-decreasing w.r.t. the iteration index, respectively. As a result, the performance upper bound and the 
performance lower bound go to converge \cite{9140412}.
Therefore, through iteratively solving primal problem and master problem, we can obtain the optimal solution when the upper bound and lower bound are sufficiently close \cite{8910634,7106496}.
The detailed algorithm is summarized in Algorithm \ref{alg_GBD}.
}

\subsection{Complexity Analysis}
{ 
The complexity of solving problem \eqref{P1} by Algorithm~\ref{alg_GBD} lies in solving the primal problem, feasibility problem, and master problem at each iteration. 
For primal problem, where we iteratively update time allocation variables and frequency variables. The frequency optimization method is given in Algorithm~\ref{alg1}, whose complexity is $\mathcal{O}\left(\frac{K+2-\varkappa}{\sqrt{\epsilon_1}}\log_2(1/\epsilon_0)\right)$ as analyzed in Section III. The time allocation optimization is according to Proposition~\ref{proposition_t}, whose complexity is estimated as $\mathcal{O}\left(K\log_2(T)\right)$. Therefore, the total complexity of solving primal problem is 
$\mathcal{O}\left(\frac{K+2-\varkappa}{\sqrt{\epsilon_1}}\log_2(1/\epsilon_0)K\log_2(T)L_1\right)$, where $L_1$ denotes the iteration number in the primal problems. 
For the feasibility problem, the complexity is given by $\mathcal{O}\left(\left(\frac{(K+1)K}{2}+3K+1\right)^{3.5}\right)$ by the interior point method. For the master problem, the computational complexity is $\mathcal{O}\left(2^K\right)$ by the Branch and Bound (BnB) method \cite{7973989}. 
}

\begin{algorithm}[t]
\begin{small} 
	\caption{Joint User Scheduling and Resource Allocation Algorithm}
	\label{alg_GBD}
	
	{\bf Initialize} arbitrary feasible user scheduling $\pmb{A}^{(j)}$, and set $j=1$, $UB=+\infty$, $LB=-\infty$, $\mathcal{J}_1=\mathcal{J}_2=\emptyset$;  
	
	\Repeat{\textnormal{$UB$ and $LB$ are sufficiently close}}{ 
		 
		\eIf{\textnormal{problem \eqref{P2.1} is feasible}}{
		
		\Repeat{\textnormal{the objective of \eqref{P2.1} converges}}{
			
			Obtain the optimal computation resource $\pmb{f}^{(j)}$ and dual variable $\beta_k^{(j)}$ according to Algorithm~\ref{alg1};
		
			Obtain the optimal time allcation $\Delta \pmb{t}^{(j)}$ and dual variable $\rho_k^{(j)}$;
			}
		
		Update $UB$ and $\mathcal{J}_1$; 
		
		}{
		Solve feasibility problem \eqref{P2_feasible} and update $\mathcal{J}_2$; 
		
		Obtain the corresponding optimal solution  $\hat{\pmb{x}}^{(j)}$ and $\Delta \hat{\pmb{t}}^{(j)}$ as well as dual variables $\hat\rho_k^{(j)}$ and $\hat\beta_K^{(j)}$; 
		}
		
		Solve master problem \eqref{master} by adding optimality cuts \eqref{opt_cut} and feasibility cuts \eqref{feas_cut}; 
		
		Set $j\leftarrow j +1$;
		
		Update $\pmb{A}^{(j)}$ and $LB$; 
	}

	Output  the optimal $\pmb f^*$, $\Delta \pmb t^*$ and $\pmb{A}^*$. 
\end{small}  
\end{algorithm}

\section{Simulations}
In this section, we perform simulations to validate the  proposed scheme and algorithm. 
There are  $K=10$ devices around the server.  
 The task size $A_k$ and computation intensity obey uniform distribution on $[10,50]$ Kbits and $[500, 1500]$ cycles/bit,  respectively.  The transmit power of BS is $P_0=3$~W.  The energy coefficient of the MEC server and energy conversion factor of devices are set as $\kappa = 10^{-26}$ and $\eta=0.51$. We set the energy constant of transmission $\lambda=10^{-25}$. Furthermore, the maximum computation resource is  $F_{\max}=1$ GHz and the allowable delay is $T=1$ second.    
 {In channel model, we set antenna gain $A=3$, carrier frequency $f_c=915$ MHz, path-loss factor $\ell=3$, speed of light $c=3\times10^8$ m/s, and the Rician factor is $\gamma=0.3$.}    
{The following benchmarking schemes are provided: 
	\begin{itemize}
		\item \emph{JSORA\cite{9632276}}: The joint sensing-and-offloading resource allocation algorithm, where the allocated frequency for each device keeps unchanged during its computation duration, i.e., $f_{n,n+1}=\cdots=f_{n,K+1}$ $(\forall n\in\mathcal{K})$.
		\item \emph{JCCRM-Sync\cite{8264794,7997360,8249785,8986845,8234686,8334188,9881553}}: The joint communication and computation resource management algorithm, where the computing of server will not begin until all tasks are received, which is adopted by most of the literature. 
		\item \emph{Random scheduling  scheme\cite{8952620,9292432}}: We randomly set the user scheduling for offloading.  
		\item \emph{Exhaustive search}: We randomly choose multiple initial points for Algorithm 2 and
		select the smallest result as output. The results of exhaustive search 
		method can be regarded as global optimal solutions. 
	\end{itemize}
 Besides, all   accuracies used in the simulations are set as $10^{-5}$ for fairness. }

\begin{figure}[t]
	\centering
	\begin{minipage}[t]{0.45\textwidth} 
		\centering
		\includegraphics[width=1\textwidth]{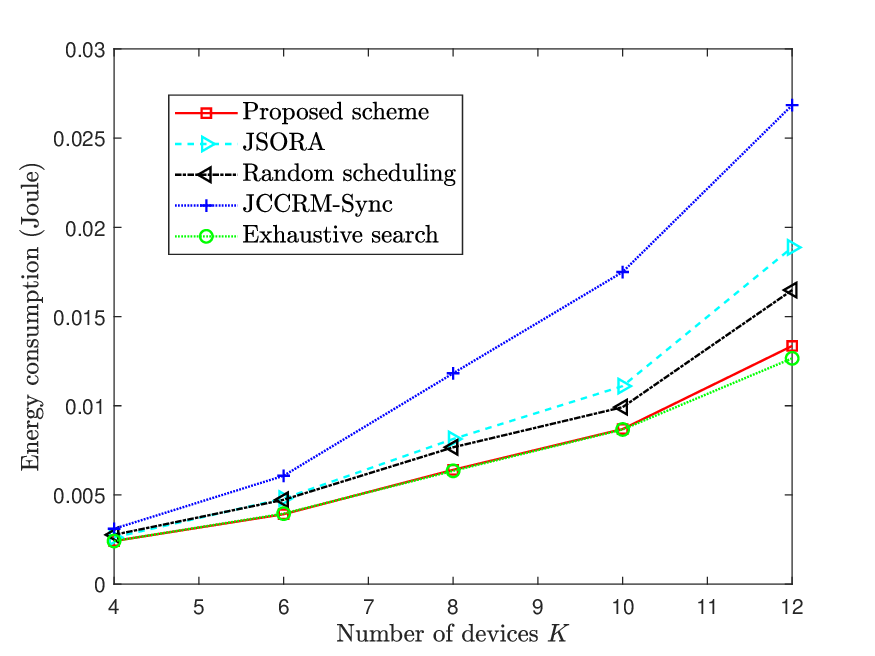} 
		\caption{Energy consumption versus number of devices.}\label{E_K} 
	\end{minipage}
	\begin{minipage}[t]{0.45\textwidth}
		\centering
		\includegraphics[width=1\textwidth]{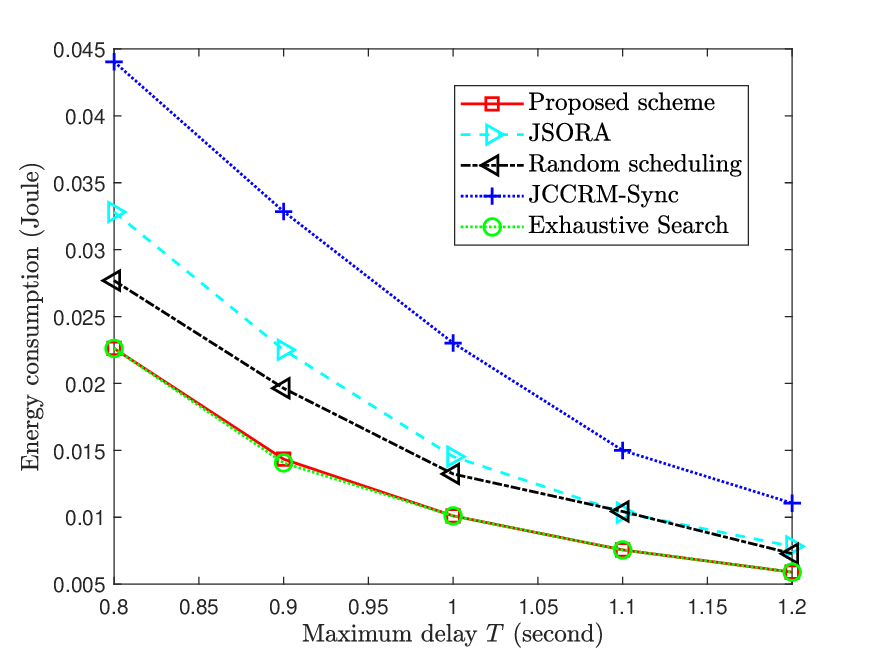}
		\caption{Energy consumption versus maximum delay.}\label{E_T}
	\end{minipage}
\end{figure}

{ 
Fig.~\ref{E_K} demonstrates the  energy consumption performance comparisons between different
schemes under different numbers of devices. We can observe that the energy consumption of the proposed scheme as well as benchmark schemes increases with the number of devices getting large. This is because that devices have to compete for fixed communication and computation resource. As the number of devices increases, the average transmission time and computation time of each device get small, thus average computation resource becomes large. Therefore, larger energy consumption of server is required in order to finish devices' tasks within the required delay.  Moreover, as can be seen in Fig.~\ref{E_K}, the gap between the proposed algorithm and exhaustive search scheme is small. This indicates that the proposed algorithm achieves close-to-optimal solutions. Compared with JSORA scheme, random scheduling scheme and JCCRM-Sync scheme, the proposed scheme achieves $30.88\%$, $19.51\%$, $87.87\%$  energy  reductions, respectively.   This can be explained by that the proposed algorithm can take full advantage of the flexibility of asynchronous computing and user scheduling. Particularly, compared with the proposed scheme, JCCRM-Sync scheme wastes the idle computation resource from time slots $t_2$ to $t_K$. Similarly, JSORA scheme can not make full use of computation resource from time slots $t_2$ to $t_K$. Hence, its performance is  better than JCCRM-Sync but worse than the proposed scheme. Additionally, random scheduling, as most of the existing literature does, can not utilize the heterogeneity of tasks size and computation intensity well in MEC networks. 
}

{ 
In Fig.~\ref{E_T}, we depict the energy consumption curves of different schemes versus the maximum allowable delay. As can be seen, the energy consumption of all schemes decreases as the maximum delay becomes large. This is because as delay gets large, the server has more time to finish tasks. Thus, the fewer  computation resource is allowable. Hence,  energy consumption can be reduced. From Fig.~\ref{E_T}, it can be  verified that the proposed algorithm outperforms benchmarking schemes in terms of energy consumption in the considered region of delay,   especially in resource-scarce scenarios. This phenomenon can be observed in Fig.~\ref{E_K} and Fig.~\ref{E_T} that the difference in energy consumption between the  proposed algorithm and benchmark schemes gets small when resource is abundant. This is because the flaws of benchmark schemes compared with the proposed algorithm can be    appropriately compensated by utilizing additional sources. Furthermore, it should be noticed  that JSORA scheme is equivalent to the proposed scheme when computation resource is abundant according to  Theorem~\ref{theorem_5}. 
}

\begin{figure} [t]
	\centering 
	\subfigure[Transition point: $t_3$]{\label{}
		\includegraphics[width=0.23\textwidth]{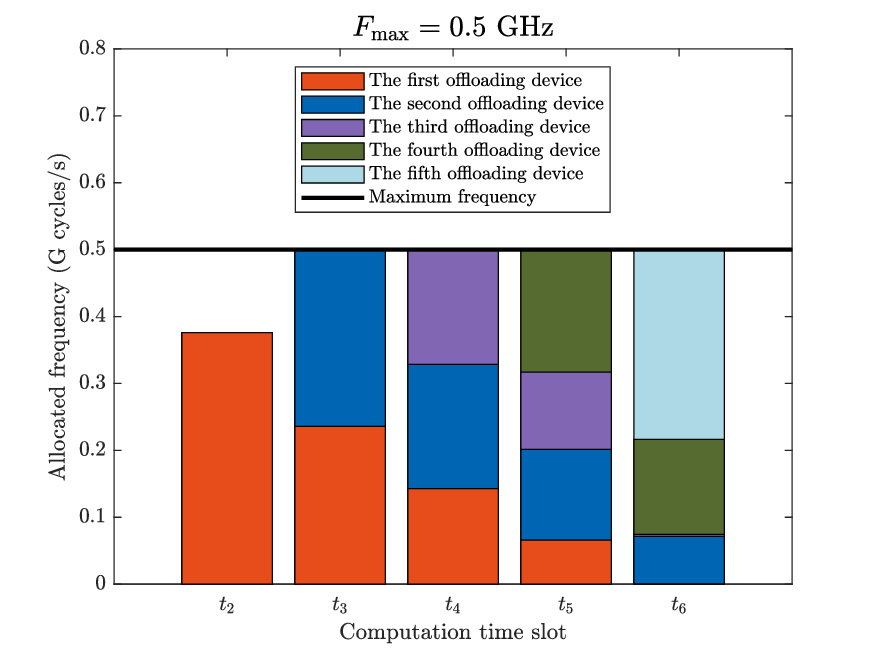}}
	\subfigure[Transition point: $t_5$]{\label{}
		\includegraphics[width=0.23\textwidth]{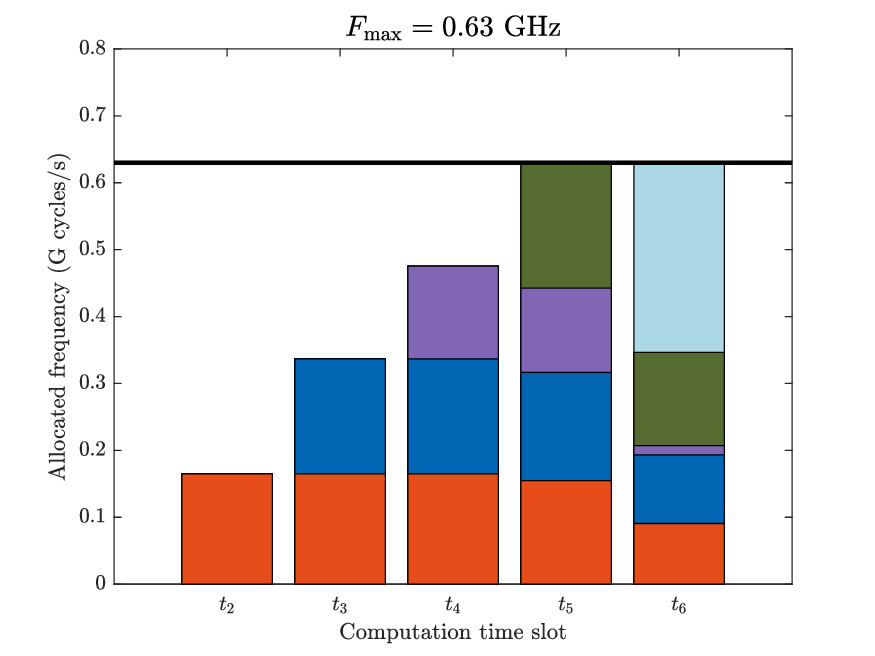}}
	
	\subfigure[Transition point: $t_6$]{\label{}
		\includegraphics[width=0.23\textwidth]{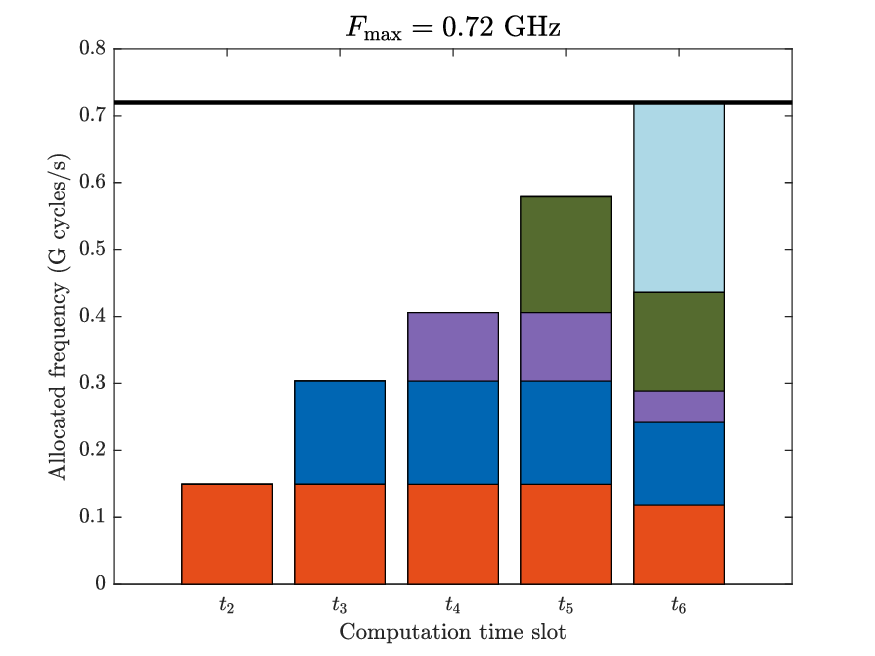}}
	\subfigure[No transition point]{\label{}
		\includegraphics[width=0.23\textwidth]{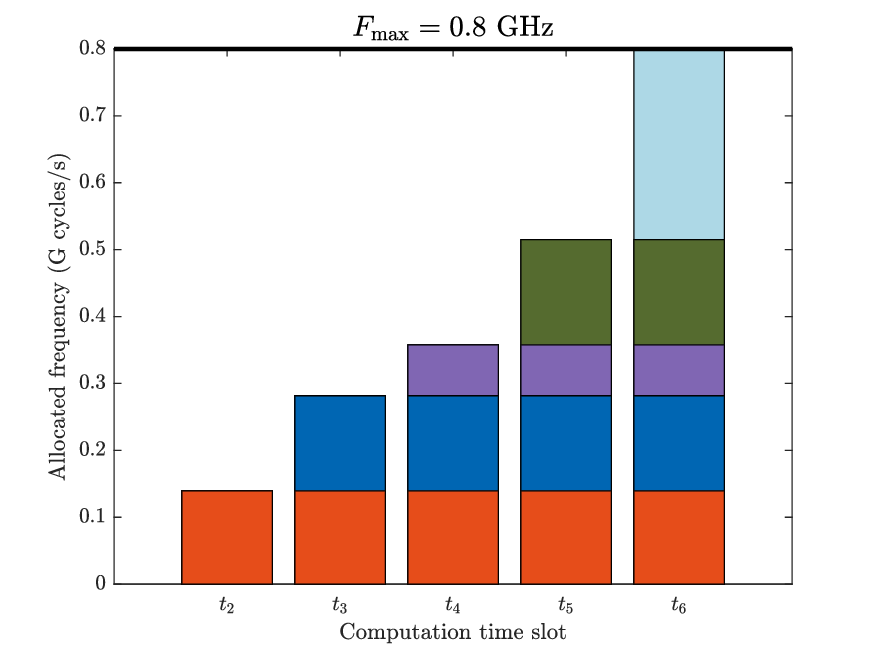}}
	\caption{Illustrations of different transition points under different maximum frequencies.}
	\label{E_bar}\vspace{-2ex}
\end{figure}

Fig.~\ref{E_bar} illustrates a specific case of the allocated frequency of each device at each computation slot  under different maximum computation frequencies when $K=5$. It can be seen that as the maximum frequency $F_{\max}$ becomes large, the transition point is gradually postponed, and finally no transition point exists when computation resource is sufficiently large which is in accordance with Theorem~\ref{theorem_5}. Specifically, for each subfigure, we can find that before the transition point, the allocated frequency for each device being computed remains unchanged and the maximum frequency constraints do not work. From the transition point to the end, the allocated frequency for each device becomes small and the maximum frequency of 
 the server is used. This verifies  Theorem~\ref{theorem_5}.  

\begin{figure}[t]
	\centering
	\includegraphics[width=0.45\textwidth]{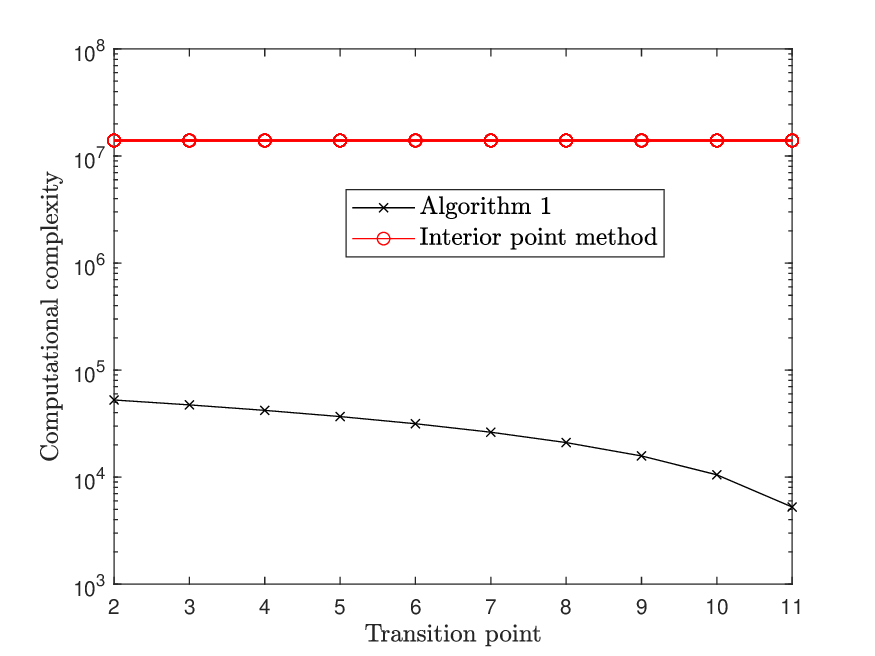}
	\caption{Comparisons of computational complexity with $K=10$. } \label{fig_complexity}
	\vspace{-1.5em}
\end{figure}

Fig.~\ref{fig_complexity} depicts the computational complexity comparisons  between the proposed Algorithm~\ref{alg1} and the interior point method under different transition points. As can be seen, the computational complexity of Algorithm~\ref{alg1} is significantly reduced compared with the interior point method, by more than $100$ times on average. As the transition point becomes larger, the complexity further decreases. For example, when the transition point $\varkappa=11$, the complexity of Algorithm~\ref{alg1} is reduced by $1000$ times. This is because the proposed computation resource allocation algorithm fully utilizes the properties in Theorem~\ref{theorem_5}
to reduce algorithm complexity, especially when the computation resource of the server is abundant. 

\begin{figure}[t]
	\centering
	\includegraphics[width=0.45\textwidth]{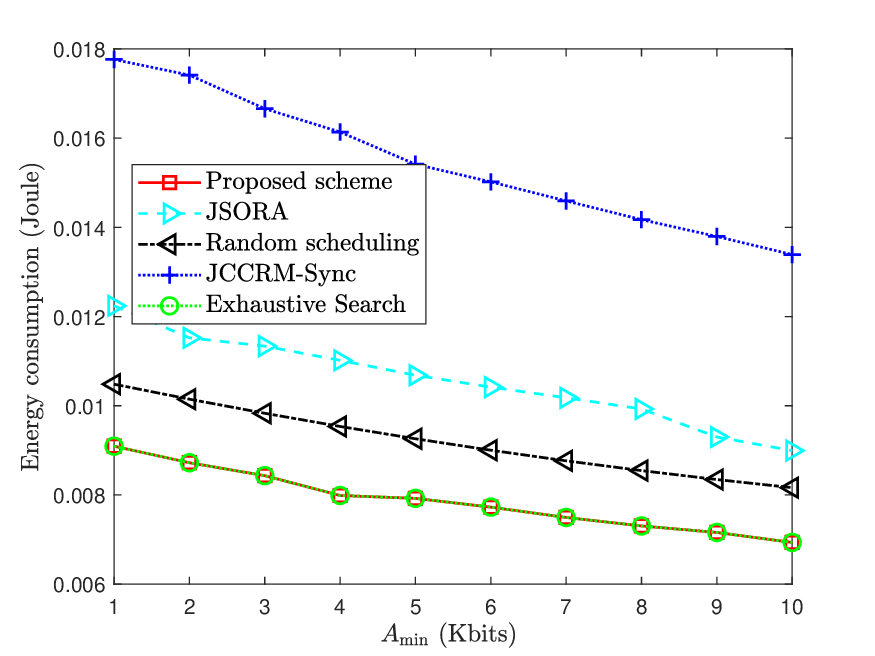}
	\caption{Energy  consumption versus minimum task size.} \label{fig_Amin}
	\vspace{-1.5em}
\end{figure}

{ 
To test the compatibility of the proposed algorithm under different task scale differences, the energy consumption versus minimum task size $A_{\min}$ is shown in Fig.~\ref{fig_Amin}, where task size obeys uniform distribution on $[A_{\min},A_{\max}]$ with fixed  mean value $\frac{A_{\min}+A_{\max}}{2}=30$ Kbits. With a large $A_{\min}$, the task scale difference is small. In Fig.~\ref{fig_Amin}, the proposed scheme and exhaustive search scheme achieve nearly the same performance,  and outperform other schemes.   
One can observe that the energy consumption increases as task scale difference gets large. This can be explained by that the resources have to be tilted towards the devices with large task sizes, thus resulting in more energy consumption.  
}

\section{Conclusion}
In this paper, we have investigated a joint user scheduling
and resource optimization framework for MEC networks with
asynchronous computing. An optimization problem of joint
user scheduling, communication and computation resource
management has been solved aiming to minimize the
energy consumption of server under the delay constraint.
Simulations verified that the proposed algorithm yields significant
performance gains compared with benchmark schemes.  This work establishes a new principle of asynchronous computing and verifies the superiority over its synchronous counterpart. 
{For future works, we will generalize the proposed  asynchronous computing framework to  heterogeneous task  deadlines scenarios so as to further activate its potential. As another direction, the extension to online algorithm design and accommodate new coming devices deserve further investigation. 
}

\appendices
\section{Proof of Proposition  \ref{proposition_1}}  \label{proof_proposition_1}
\setcounter{equation}{0}
\renewcommand\theequation{A.\arabic{equation}}
The feasibility problem of \eqref{theorem4_eq1} can be constructed as
\begin{subequations} \label{proposition_eq1}
	\begin{align}
		\min_{f_{n,m}}\quad &\max_{m \in \{2, \cdots, K+1\}} \sum_{n=1}^{m-1} f_{n,m},\tag{\ref{proposition_eq1}}\\
		s.t.\quad  
		&\sum_{m=n+1}^{K+1} f_{n,m}\Delta t_m \geq F_n, \quad \forall n\in\mathcal{K},\\
		&f_{n,m}\geq 0,\quad\forall n\in\mathcal{K},\forall m = n+1, \cdots, K+1.
	\end{align}
\end{subequations}
If the optimal objective of problem  \eqref{proposition_eq1} is less than or equal to $F_{\max}$, problem \eqref{theorem4_eq1} is feasible; otherwise, it is infeasible. Subsequently, we analyze the optimal solution of problem \eqref{proposition_eq1}. First, when $K=1$, i.e., there exists only one task, the optimal objective of problem  \eqref{proposition_eq1} is $\frac{F_K}{\Delta t_{K+1}}$. When $K=2$,  we consider two cases: 1) If $\frac{F_1}{\Delta t_2}\leq\frac{F_2}{\Delta t_3}$, this indicates that the optimal scheme is computing task $2$ after task $1$ is finished. Therefore, the optimal solution is $f_{1,2}=\frac{F_1}{\Delta t_2}$ and $f_{2,3}=\frac{F_2}{\Delta t_3}$. Hence, the optimal objective is  $\frac{F_2}{\Delta t_{3}}$. 2) If $\frac{F_1}{\Delta t_2}>\frac{F_2}{\Delta t_3}$, this implies that part of task $1$ can be processed in parallel with task $2$. Hence, the optimal solution is given by $f_{1,2}=f_{1,3}+f_{2,3}=\frac{F_1+F_2}{\Delta t_2+\Delta t_3}$. Since $\frac{F_1+F_2}{\Delta t_2+\Delta t_3}>\frac{F_2}{\Delta t_3}$, the optimal objective is $\frac{F_1+F_2}{\Delta t_2+\Delta t_3}$. In conclusion, when $K=2$, the optimal objective is $\max\{\frac{F_1+F_2}{\Delta t_2+\Delta t_3},\frac{F_2}{\Delta t_3}\}$. Similarly, by recursion, we can deduce that when there exist $K$ TD, the optimal solution is $\max\{\frac{F_1+\cdots+F_K}{\Delta t_2+\cdots+\Delta t_{K+1}},\frac{F_2+\cdots+F_K}{\Delta t_3+\cdots+\Delta t_{K+1}},\cdots,\frac{F_K}{\Delta t_{K+1}}\}$. That completes the proof.    \hfill $\blacksquare$

\section{Proof of Property 1) in Theorem \ref{theorem_5}}  \label{proof_theorem_5}
\setcounter{equation}{0}
\renewcommand\theequation{B.\arabic{equation}} 	
Before that, we give the following two corollaries to facilitate the proof. 
\begin{coro} \label{theorem1}
	\emph{[Row property] The optimal computation resource of each task is non-increasing during its computation period, i.e., $f_{n,n+1}^*\geq f_{n,n+2}^*\geq\cdots\geq f_{n,K+1}^*$ $(\forall n \in \mathcal{K})$.}
\end{coro}

\begin{proof}
	Please refer to Appendix \ref{proof_theorem1}. 
\end{proof}


\begin{coro} \label{coro_2}
	\emph{[Column property] Denote the sum computation cycles of the $n$-th offloading device in time slots $t_m$ and $t_{m+1}$ by $F_n$ for  $m=3,\cdots,K$ and $n=1,\cdots,m-1$. If $F_n>0$ holds for all $n=1,\cdots,m-1$, the optimal frequency shifts  $\delta^*_n\triangleq f^*_{n,m}-f^*_{n,m+1}$ $(n=1,\cdots,m-1)$ are either all zeros or all positive, i.e., have the coincident zero or positive  characteristics.    } 
\end{coro}
\begin{proof}
	Please refer to Appendix \ref{proof_coro_2}. 
\end{proof}
	First, applying the KKT conditions gives
	\begin{align}
		&3\kappa (f^*_{n,m})^2 \Delta t_m+\alpha_m-\beta_n\Delta t_m-\gamma_{n,m}=0,\label{theorem4_eq3}\\
		&\alpha_m\left(\sum_{n=1}^{m-1} f^*_{n,m}- F_{\max}\right)=0,\label{theorem4_eq4}\\
		&\beta_n\left(F_n-\sum_{m=n+1}^{K+1} f^*_{n,m}\Delta t_m \right)=0,\\
		&\gamma_{n,m}f^*_{n,m}=0,\quad\forall n\in\mathcal{K},\label{theorem4_eq6}\\
		&\alpha_m\!\geq\!0,\beta_n\!\geq\!0,\gamma_{n,m}\!\geq\!0.
	\end{align}
	Based on \eqref{theorem4_eq3}, we obtain that
	\begin{align} \label{theorem4_eq8}
		f^*_{n,m}=\sqrt{\frac{\beta_n}{3\kappa}+\frac{\gamma_{n,m}-\alpha_m}{3\kappa \Delta t_m}}. 
	\end{align}
	
	In case of $f_{n,m}> 0$ $(\forall n\in\mathcal{K},\forall m = n+1, \cdots, K+1)$, we have $\gamma_{n,m}=0$ according to \eqref{theorem4_eq6}. 
	Furthermore, $\beta_n>0$ is derived from \eqref{theorem4_eq8}. According to Corollary \ref{theorem1}, the optimal solution satisfies $f^*_{n,m}\geq f^*_{n,m+1}$. Thus, we have $\frac{\alpha_m}{\Delta t_m}\leq\frac{\alpha_{m+1}}{\Delta t_{m+1}}$. 
	Assume that there exists a certain  $i\in\{3,\cdots,K\}$ such that $f^*_{n,i}=f^*_{n,i+1}$. We have $\frac{\alpha_i}{\Delta t_i}=\frac{\alpha_{i+1}}{\Delta t_{i+1}}$. If $\frac{\alpha_i}{\Delta t_i}=\frac{\alpha_{i+1}}{\Delta t_{i+1}}>0$, i.e., $\alpha_i>0$ and $\alpha_{i+1}>0$, we should have $\sum_{n=1}^{i-1} f_{n,i}=\sum_{n=1}^{i} f_{n,i+1}= F_{\max}$ according to \eqref{theorem4_eq4}. 
	Furthermore, due to $f_{n,m}>0$ $(\forall n \in\mathcal{K},\forall m\in=n+1,\cdots,K+1)$, the computation cycles $f_{n,i}\Delta t_i+f_{n,i+1}\Delta t_{i+1}>0$ $(\forall n \in 1,\cdots,i-1)$.  According to Corollary \ref{coro_2}, we have $f_{n,i}=f_{n,i+1}$ $(n=1,\cdots,i-1)$. 
	Hence, it can be derived that $f_{i,i+1}=0$ which contradicts that $f_{i,i+1}$ is positive. Therefore, we have $\alpha_i=\alpha_{i+1}=0$, i.e., $\frac{\alpha_i}{\Delta t_i}=\frac{\alpha_{i+1}}{\Delta t_{i+1}}=0$. Since $\frac{\alpha_2}{\Delta t_2}\leq\cdots\leq\frac{\alpha_i}{\Delta t_i}$ and $\alpha_2,\cdots,\alpha_{i-1}\geq0$, we can further obtain that $\alpha_2=\cdots=\alpha_{i+1}=0$.  This indicates that if there exists a certain $i\in\{3,\cdots,K\}$ such that  $f^*_{n,i}=f^*_{n,i+1}>0$, we have $f^*_{n,n+1}=f^*_{n,n+2}=\cdots=f^*_{n,i+1}$. 
	
	Additionally, if there exists a certain $i< j\leq K+1$ such that $f^*_{n,j}=0$, we can deduce that $f^*_{n,j}=f^*_{n,j+1}=\cdots= f_{n,K+1}^*=0$ since $f^*_{n,j}\geq f^*_{n,j+1}\geq \cdots\geq f^*_{n,K+1}$. 
	
	Combing the above two cases, we complete the proof.  \hfill$\blacksquare$

\section{Proof of Property 2) in Theorem  \ref{theorem_5}}  \label{proof_theorem1_2)}
\setcounter{equation}{0}
\renewcommand\theequation{C.\arabic{equation}} 
According to Appendix~\ref{proof_theorem_5}, we can obtain that  $\alpha^*_2=\cdots=\alpha^*_i=0$.
	Moreover, since  $f^*_{n-1,n}>f^*_{n-1,n+1}\geq 0$  $(i\leq n \leq K)$, we have $\sqrt{\frac{\beta_{n-1}}{3\kappa}+\frac{\gamma_{n-1,n}-\alpha_n}{3\kappa \Delta t_n}}>\sqrt{\frac{\beta_n-1}{3\kappa}+\frac{\gamma_{n-1,n+1}-\alpha_{n+1}}{3\kappa \Delta t_{n+1}}}$ according to \eqref{theorem4_eq8}.   Due to that $\gamma_{n-1,n}=0$ and $\gamma_{n-1,n+1}>0$, we can deduce that $-\frac{\alpha_n}{\Delta t_n}>\frac{\gamma_{n-1,n+1}-\alpha_{n+1}}{\Delta t_{n+1}}>-\frac{\alpha_{n+1}}{\Delta t_{n+1}}$, i.e., $\frac{\alpha_n}{\Delta t_n}<\frac{\alpha_{n+1}}{\Delta t_{n+1}}$ $(i\leq n \leq K)$.   Overall, we can conclude that $0=\frac{\alpha^*_2}{\Delta t_2}=\cdots=\frac{\alpha^*_i}{\Delta t_i}<\frac{\alpha^*_{i+1}}{\Delta t_{i+1}}\cdots<\frac{\alpha^*_{K+1}}{\Delta t_{K+1}}$. \hfill$\blacksquare$

\section{Proof of Property 3) in Theorem \ref{theorem_5}}  \label{proof_theorem_8}
\setcounter{equation}{0}
\renewcommand\theequation{D.\arabic{equation}} 	
We first prove the ``only if'' part. According to property 1) in Theorem~\ref{theorem_5}, if $t_{i}$ is the transition point, we have $f_{n,n+1}^*=\cdots= f_{n,i-1}^*>f_{n,i}^*\cdots>f^*_{n,j}=\cdots= f_{n,K+1}^*=0$  $(n+1\leq i < j\leq K+1)$ and $\sum_{n=1}^{m-1}f_{n,m}=F_{\max}$ $(i\leq m\leq K+1)$. Since $\sum_{m=n+1}^{K+1}f^*_{n,m}\Delta t_{m}=F_n$, we can obtain that  $\sum_{m=i}^{K+1}f^*_{n,m}\Delta t_m<\frac{F_n(\sum_{m=i}^{K+1}\Delta t_m)}{\sum_{m=n+1}^{K+1}\Delta t_m}$ $(1\leq n\leq i-2)$. Thus, we have $(\sum_{m=i}^{K+1}\Delta t_m)F_{\max}=\sum_{n=1}^{i-2}\sum_{m=i}^{K+1}f^*_{n,m}\Delta t_m +\sum_{n=i-1}^{K}\sum_{m=n+1}^{K+1}f^*_{n,m}\Delta t_m <\sum_{n=1}^{i-2}\frac{F_n(\sum_{m=i}^{K+1}\Delta t_m)}{\sum_{m=n+1}^{K+1}\Delta t_m} + \sum_{n=i-1}^K F_n$, i.e., $F_{\max}<\digamma(i)$. 

Similarly, if $t_{i-1}$ is the transition point, we have  $F_{\max}<\digamma(i-1)$. Since $ \digamma(i)-\digamma(i-1) = \frac{F_{i-2}}{\sum_{m=i-1}^{K+1}\Delta t_m} + \frac{\sum_{n=i-1}^K F_n}{\sum_{m=i}^{K+1}\Delta t_m} - \frac{\sum_{n=i-2}^K F_n}{\sum_{m=i-1}^{K+1}\Delta t_m} = \frac{\sum_{n=i-1}^K F_n}{\sum_{m=i}^{K+1}\Delta t_m} - \frac{\sum_{n=i-1}^K F_n}{\sum_{m=i-1}^{K+1}\Delta t_m}>0$, we should have $\digamma(i-1) \leq F_{\max}<\digamma(i)$.

For the ``if'' part, if $F_{\max}<\digamma(i)$, we can deduce that $(\sum_{m=i}^{K+1}\Delta t_m)F_{\max} <\sum_{n=1}^{i-2}\frac{F_n(\sum_{m=i}^{K+1}\Delta t_m)}{\sum_{m=n+1}^{K+1}\Delta t_m} + \sum_{n=i-1}^K F_n$. Moreover, since $\sum_{n=1}^{i-2}\sum_{m=i}^{K+1}f^*_{n,m}\Delta t_m + \sum_{n=i-1}^K F_n\leq(\sum_{m=i}^{K+1}\Delta t_m)F_{\max}$, we have $\sum_{n=1}^{i-2}\sum_{m=i}^{K+1}f^*_{n,m}\Delta t_m + \sum_{n=i-1}^K F_n<\sum_{n=1}^{i-2}\frac{F_n(\sum_{m=i}^{K+1}\Delta t_m)}{\sum_{m=n+1}^{K+1}\Delta t_m} + \sum_{n=i-1}^K F_n$, i.e., $\sum_{n=1}^{i-2}\sum_{m=i}^{K+1}f^*_{n,m}\Delta t_m <\sum_{n=1}^{i-2}\frac{F_n(\sum_{m=i}^{K+1}\Delta t_m)}{\sum_{m=n+1}^{K+1}\Delta t_m}$. Since $f_{n,n+1}^*\geq \cdots\geq f_{n,i-1}^*\geq\cdots\geq f_{n,K+1}^*$  $(n+1\leq i <  K+1)$, we can obtain that $\sum_{m=i}^{K+1}f^*_{n,m}\Delta t_m\leq\frac{F_n(\sum_{m=i}^{K+1}\Delta t_m)}{\sum_{m=n+1}^{K+1}\Delta t_m}$ $(1\leq n \leq i-2)$. Therefore, to let  $\sum_{n=1}^{i-2}\sum_{m=i}^{K+1}f^*_{n,m}\Delta t_m <\sum_{n=1}^{i-2}\frac{F_n(\sum_{m=i}^{K+1}\Delta t_m)}{\sum_{m=n+1}^{K+1}\Delta t_m}$ hold, we should have $\sum_{m=i}^{K+1}f^*_{n,m}\Delta t_m<\frac{F_n(\sum_{m=i}^{K+1}\Delta t_m)}{\sum_{m=n+1}^{K+1}\Delta t_m}$ $(1\leq n \leq i-2)$. Further, it can be deduced that $\sum_{m=i-1}^{K+1}f^*_{n,m}\Delta t_m<\frac{F_n(\sum_{m=i-1}^{K+1}\Delta t_m)}{\sum_{m=n+1}^{K+1}\Delta t_m}$ $(1\leq n \leq i-3)$.

Additionally, since  $\digamma(i-1) \leq F_{\max}$, we can deduce that $ \sum_{n=1}^{i-3}\frac{F_n(\sum_{m=i-1}^{K+1}\Delta t_m)}{\sum_{m=n+1}^{K+1}\Delta t_m} + \sum_{n=i-2}^K F_n \leq (\sum_{m=i-1}^{K+1}\Delta t_m)F_{\max}$. Thus, we have
$ \sum_{n=1}^{i-3}\sum_{m=i-1}^{K+1}f^*_{n,m}\Delta t_m+ \sum_{n=i-2}^K F_n <\sum_{n=1}^{i-3}\frac{F_n(\sum_{m=i-1}^{K+1}\Delta t_m)}{\sum_{m=n+1}^{K+1}\Delta t_m} + \sum_{n=i-2}^K F_n \leq (\sum_{m=i-1}^{K+1}\Delta t_m)F_{\max}$, which indicates that the computation resource is abundant  from $t_{i-1}$ to $t_{K+1}$. Therefore, we can deduce that $f^*_{n,n+1}=\cdots=f^*_{n,i-1}$ $(n \leq i-2)$. 
Assume that $t_{\tilde i}$ $(\tilde i > i)$ is the transition point. We have $\digamma(\tilde i -1) \leq F_{\max}$. Since $\digamma(\tilde i -1) - \digamma(i) \geq 0$, $F_{\max}$ is infeasible, which breaks the assumption. Therefore, we can conclude that $t_{i}$ is the transition point. 

Combining the proofs of ``if'' and ``only if'' part, we complete the proof. \hfill$\blacksquare$

{
\section{Proof of Corollary   \ref{theorem1}}  \label{proof_theorem1}
\setcounter{equation}{0}
\renewcommand\theequation{E.\arabic{equation}} 
 To find out the optimal computation resource  allocation scheme, we first investigate the property of the most energy-efficient scheme without the maximum frequency restriction in Lemma~\ref{lemma1}, whose proof is provided in Appendix~\ref{proof_lemma1}. 

\begin{lemma} \label{lemma1}
	\emph{Regardless of  $\Delta t_{m}$,  $\Delta t_{m+1}$ and with given computation cycles $F_n$ in time slots $t_m$ and $t_{m+1}$,   scheme $f^*_{n,m}=f^*_{n,m+1}$ consumes the least energy among all the solutions satisfying 
	$f_{n,m}\Delta t_m+ f_{n,m+1}\Delta t_{m+1}=F_n$. 
 }
\end{lemma}

For Corollary~\ref{theorem1}, we first prove that $f_{1,2}^*\geq f_{1,3}^*$ with given computation cycle $F_1$ and $f_{2,3}$.  Denote the sum computation cycles in $t_2$ and $t_3$ of the first offloading device as $F_1$. Through relaxing the maximum computation resource constraint, the energy consumption is the least when $f_{1,2}=f_{1,3}=\frac{F_1}{\Delta t_2+\Delta t_3}$ according to Lemma  \ref{lemma1}. 
Since constraint (\ref{theorem4_eq1}a) should be satisfied, we have 
\begin{align} \label{lemma1_eq3}
	f_{1,2}\leq F_{\max}, \quad
	f_{1,3}+f_{2,3}\leq F_{\max}.  
\end{align}
We consider the following two cases; otherwise, $f_{1,2}$ and $f_{1,3}$ have no feasible solution with the given $F_1$. 

\textit{Case 1: $\frac{F_1}{\Delta t_2+\Delta t_3}+f_{2,3}\leq F_{\max}$.} In this case, we can deduce that $f_{1,2}=f_{1,3}=\frac{F_1}{\Delta t_2+\Delta t_3}$ satisfies \eqref{lemma1_eq3}. Since $f_{1,2}=f_{1,3}=\frac{F_1}{\Delta t_2+\Delta t_3}$ is the most energy efficient solution, the optimal solution in this case is  $f^*_{1,2}=f^*_{1,3}=\frac{F_1}{\Delta t_2+\Delta t_3}$.

\textit{Case 2: $\frac{F_1}{\Delta t_2+\Delta t_3}+f_{2,3}> F_{\max}$.} Obviously,  $f_{1,2}=f_{1,3}=\frac{F_1}{\Delta t_2+\Delta t_3}$ is infeasible in this case. We then prove that $f_{1,2}<f_{1,3}$ is also impossible. Since $f_{1,2}\Delta t_2+f_{1,3}\Delta t_3=F_1$, $f_{1,2}$ increases as $f_{1,3}$ decreases.  If $f_{1,2}<f_{1,3}$, we can deduce that $f_{1,3}>\frac{F_1}{\Delta t_2+\Delta t_3}$. Therefore, we have  $f_{1,3}+f_{2,3}>\frac{F_1}{\Delta t_2+\Delta t_3}+f_{2,3}>F_{\max}$ which violates the maximum frequency constraint. As a consequence, the optimal solution is  $f_{1,2}^*>f_{1,3}^*$.     According to Lemma \ref{lemma1}, the energy consumption $E$ increases with $\delta_1=f_{1,2}-f_{1,3}$ in the considered region $0<\delta_1\leq F_1/\Delta t_2$. Therefore, in order to achieve the fewest energy consumption, we should let $\delta_1$ as small as possible. Hence, we can obtain that  $f^*_{1,3}=F_{\max}-f_{2,3}$ and  $f^*_{1,2}=\frac{F_1-f^*_{1,3}\Delta t_3}{\Delta t_2}$. The corresponding optimal $\delta_1^*=\frac{F_1-(F_{\max}-f_{2,3})(\Delta t_2+\Delta t_3)}{\Delta t_2}$.

Summarizing the above two cases, we can obtain that $f^*_{1,2}\geq f^*_{1,3}$. Subsequently, we prove that in $t_m$ and $t_{m+1}$ $(m=3,\cdots,K)$, we always have $f^*_{n,m}\geq f^*_{n,m+1}$ for all $n=1,\cdots,m-1$. 

Denote the sum computation cycles of the $n$-th offloading device in time slots $t_m$ and $t_{m+1}$ by $F_n$, i.e.,  
\begin{align} \label{lemma1_3.5}
	\left\{\begin{aligned}
		&f_{1,m}\Delta t_m+f_{1,m+1}\Delta t_{m+1}=F_1,  \\		&\quad\quad\quad\quad\quad\cdots\\
		&f_{n,m}\Delta t_m+f_{n,m+1}\Delta t_{m+1}=F_n,  \\
		&\quad\quad\quad\quad\quad\cdots\\
		&f_{m-1,m}\Delta t_m+f_{m-1,m+1}\Delta t_{m+1}=F_{m-1}.
	\end{aligned}
	\right.
\end{align}
According to Lemma \ref{lemma1}, when $f_{n,m}=f_{n,m+1}=\frac{F_n}{\Delta t_m+\Delta t_{m+1}}$ for all $n=1,\cdots,m-1$, the minimum energy consumption of the $n$-th offloading device can be achieved, thus the total energy consumption is minimum. Moreover, the following constraints should be satisfied:
\begin{align} \label{lemma1_eq4}
	\sum_{n=1}^{m-1}f_{n,m}\leq F_{\max}, \quad
	\sum_{n=1}^{m}f_{n,m+1}\leq F_{\max}.
\end{align}
We consider two cases. 

\textit{Case 1: $\sum_{n=1}^{m-1} \frac{F_n}{\Delta t_m+\Delta t_{m+1}} + f_{m,m+1}\leq F_{\max}$.} In this case, we can deduce that $f_{n,m}=f_{n,m+1}=\frac{F_n}{\Delta t_m+\Delta t_{m+1}}$ for all $n=1,\cdots,m-1$ satisfies \eqref{lemma1_eq4}. Thus, the optimal solution in this case is $f^*_{n,m}=f^*_{n,m+1}=\frac{F_n}{\Delta t_m+\Delta t_{m+1}}$ for all $n=1,\cdots,m-1$. 

\textit{Case 2: $\sum_{n=1}^{m-1} \frac{F_n}{\Delta t_m+\Delta t_{m+1}} + f_{m,m+1}>  F_{\max}$.} Obviously,  $f_{n,m}=f_{n,m+1}=\frac{F_n}{\Delta t_m+\Delta t_{m+1}}$ is infeasible in this case. We then prove that $\sum_{n=1}^{m-1}f^*_{n,m}> \sum_{n=1}^{m-1}f^*_{n,m+1}$ by contradiction. By summing all the equalities in \eqref{lemma1_3.5}, we have $\Delta t_m\sum_{n=1}^{m-1}f_{n,m}+\Delta t_{m+1}\sum_{n=1}^{m-1}f_{n,m+1}=\sum_{n=1}^{m-1}F_n$. Thus, $\sum_{n=1}^{m-1}f_{n,m}$ is negatively correlated with $\sum_{n=1}^{m-1}f_{n,m+1}$. If $\sum_{n=1}^{m-1}f_{n,m}< \sum_{n=1}^{m-1}f_{n,m+1}$, it can be inferred that  $\sum_{n=1}^{m-1}f_{n,m+1}>\sum_{n=1}^{m-1}\frac{F_n}{\Delta t_m+\Delta t_{m+1}}$. We can further have $\sum_{n=1}^{m-1}f_{n,m+1}+f_{m,m+1}>\sum_{n=1}^{m-1}\frac{F_n}{\Delta t_m+\Delta t_{m+1}}+f_{m,m+1}>F_{max}$ which violates constraint \eqref{lemma1_eq4}. Similarly, if $\sum_{n=1}^{m-1}f_{n,m}= \sum_{n=1}^{m-1}f_{n,m+1}$, we can obtain that $\sum_{n=1}^{m-1}f_{n,m}= \sum_{n=1}^{m-1}f_{n,m+1}=\sum_{n=1}^{m-1}\frac{F_n}{\Delta t_m+\Delta t_{m+1}}$. Hence, we have $\sum_{n=1}^{m-1}f_{n,m+1}+f_{m,m+1}=\sum_{n=1}^{m-1}\frac{F_n}{\Delta t_m+\Delta t_{m+1}}+f_{m,m+1}>F_{max}$ which breaks constraint \eqref{lemma1_eq4}.
As a consequence, the optimal solution in this case satisfies $\sum_{n=1}^{m-1}f^*_{n,m}> \sum_{n=1}^{m-1}f^*_{n,m+1}$.

Next, we prove that $f^*_{n,m}\geq f^*_{n,m+1}$ for all $n=1,\cdots,m-1$.  With given $\Delta t_m$ and $\Delta t_{m+1}$, we denote the energy consumption of the $n$-th offloading device by $E_n(F_n,\delta_{n})$, where $\delta_{n}=f_{n,m}-f_{n,m+1}$. According to \eqref{lemma1_eq1.2},  $E_n(F_n,\delta_{n})$ is expressed by
\begin{align}
	&E_n(F_n,\delta_{n})=\kappa\frac{(F_n+\delta_{n}\Delta t_{m+1})^3}{(\Delta t_m+\Delta t_{m+1})^3}\Delta t_m\nonumber\\
 &+\kappa\frac{(F_n-\delta_{n}\Delta t_{m})^3}{(\Delta t_m+\Delta t_{m+1})^3}\Delta t_{m+1},  (n=1,\cdots,m-1),
\end{align}  
which decreases when $-F_n/\Delta t_{m+1}\leq\delta_{n}\leq0$ while increases when $0<\delta_{n}\leq F_n/\Delta t_m$. 

Furthermore, since $\sum_{n=1}^{m-1}f_{n,m+1}\leq F_{\max}-f_{m,m+1}$, we have 
\begin{align}\label{lemma1_eq5}
	&\sum_{n=1}^{m-1}f_{n,m}-\sum_{n=1}^{m-1}f_{n,m+1}\nonumber\\
 &\geq \frac{\sum_{n=1}^{m-1}F_n-\Delta t_{m+1}\sum_{n=1}^{m-1}f_{n,m+1}}{\Delta t_m}-\sum_{n=1}^{m-1}f_{n,m+1}\nonumber\\
	&\geq \frac{\sum_{n=1}^{m-1}F_n-(\Delta t_m+\Delta t_{m+1})\left(F_{\max}-f_{m,m+1}\right)}{\Delta t_m}.
\end{align}
According to \eqref{lemma1_eq5}, we have $\sum_{n=1}^{m-1}\delta_{n}\geq \frac{\sum_{n=1}^{m-1}F_n-(\Delta t_m+\Delta t_{m+1})\left(F_{\max}-f_{m,m+1}\right)}{\Delta t_m}\triangleq\Omega$. Next, we utilize contradiction to prove that $\delta_{n}^*\geq0$ for all $n=1,\cdots,m-1$. Assume in the optimal solution $\delta^*_{n}$ there exists a certain $\delta_{n}<0$. We can suitably decrease other positive $\delta^*_{n}$ and increase the negative $\delta_{n}<0$ to zero while keeping $\sum_{n=1}^{m-1}\delta_{n}$ unchanged. In this case, the total energy consumption is effectively reduced, which contradicts the optimality. That completes the proof of $\delta^*_{n}\geq0$, i.e., $f^*_{n,m}\geq f^*_{n,m+1}$ for all $n=1,\cdots,m-1$. 

In summary, since we have proven $f_{1,2}^*\geq f_{1,3}^*$ and $f_{n,m}^*\geq f_{n,m+1}^*$ for $m=3,\cdots,K$ and $n=1,\cdots,m-1$, we can deduce Corollary~\ref{theorem1}.  \hfill$\blacksquare$

	\section{Proof of Corollary  \ref{coro_2}}  \label{proof_coro_2}
	\setcounter{equation}{0}
	\renewcommand\theequation{F.\arabic{equation}} 
	In case of $\sum_{n=1}^{m-1} \frac{F_n}{\Delta t_m+\Delta t_{m+1}} + f_{m,m+1}\leq F_{\max}$, the optimal $\delta^*_n$ $(n=1,\cdots,m-1)$ are all zeros according to Corollary~\ref{theorem1}. Therefore, we only need to justify the case of $\sum_{n=1}^{m-1} \frac{F_n}{\Delta t_m+\Delta t_{m+1}} + f_{m,m+1}> F_{\max}$. 
	In this case, we first prove that the optimal $\sum_{n=1}^{m-1}\delta^*_{n}=\Omega$. Assume that the optimal $\sum_{n=1}^{m-1}\delta^*_{n}>\Omega$. We can suitably reduce the positive $\delta^*_{n}$ such that the energy consumption is further reduced, which contradicts the optimality. Therefore,  we can construct the following energy consumption minimization problem:
	\begin{subequations} \label{coro_P1}
		\begin{align}
			\min_{\pmb{\delta}} \quad & \sum_{n=1}^{m-1} E_n(F_n,\delta_{n}),\tag{\ref{coro_P1}}\\
			s.t.\quad 
			& \sum_{n=1}^{m-1} \delta_{n} = \Omega, \\
			&0\leq \delta_{n} \leq F_n/\Delta t_m,\quad \forall n = 1,\cdots,m-1,
		\end{align}
	\end{subequations}
	where $\pmb{\delta}=[\delta_{1},\cdots,\delta_{m-1}]^T$.

	Based on \eqref{lemma1_eq2}, the second derivative of $E_n(F_n,\delta_{n})$ with respect to $\delta_{n}$ is given by
	\begin{align}
		\frac{\mathrm{d}^2E_n(F_n,\delta_{n})}{\mathrm{d}(\delta_{n})^2}&\!=\!\frac{6\kappa \Delta t_m\Delta t_{m+1}}{(\Delta t_m+\Delta t_{m+1})^2}\left[\delta_{n}(\Delta t_{m+1}\!-\!\Delta t_m)\!+\!F_n\right]. 
	\end{align}
	We can infer that the second derivative of $E_n(F_n,\delta_{n})$ is always positive in the considered region 
	$-F_n/\Delta t_{m+1}\leq\delta_{n}\leq F_n/\Delta t_m$, no matter $\Delta t_m$ is larger than or smaller than, or equal to $\Delta t_{m+1}$. Hence,  $E_n(F_n,\delta_{n})$ is convex with respect to $\delta_{n}$. Thus, problem \eqref{coro_P1} is convex. The partial Lagrangian function of this problem is expressed as
	\begin{subequations} \label{coro_P2}
		\begin{align}
			\min_{\pmb{\delta}} \quad & \sum_{n=1}^{m-1} E_n(F_n,\delta_{n}) + \Upsilon\left(\sum_{n=1}^{m-1} \delta_{n} - \Omega\right),\tag{\ref{coro_P2}}\\
			s.t.\quad 
			&0\leq \delta_{n} \leq F_n/\Delta t_m,\quad \forall n = 1,\cdots,m-1,
		\end{align}
	\end{subequations}
	where $\Upsilon$ is the dual variable with respect to constraint (\ref{coro_P1}a). Problem \eqref{coro_P2} can be decomposed into a series of $(m-1)$ parallel problems: 
	\begin{subequations} \label{coro_P3}
		\begin{align}
			\min_{\delta_{n}} \quad &  E_n(F_n,\delta_{n}) + \Upsilon  \delta_{n},\tag{\ref{coro_P3}}\\
			s.t.\quad 
			&0\leq \delta_{n} \leq F_n/\Delta t_m,
		\end{align}
	\end{subequations}
	Denote the objective of \eqref{coro_P3} by $J_n$. Taking the derivative of $J_n$ with respect to $\delta_{n}$, we have 
	\begin{align} \label{coro1_eq1}
		&\frac{\mathrm{d} J_n}{\mathrm{d} \delta_{n}}=\frac{\mathrm{d} E_n(F_n,\delta_{n})}{\mathrm{d}  \delta_{n}}+\Upsilon,\nonumber\\
		&=\frac{3\kappa \Delta t_m\Delta t_{m+1}}{(\Delta t_m+\Delta t_{m+1})^2}\delta_{n}\left(\delta_{n}(\Delta t_{m+1}\!-\!\Delta t_m)\!+\!2F_n\right)\!+\!\Upsilon.
	\end{align}
	It can be deduced that $\frac{\mathrm{d} E_n(F_n,\delta_{n})}{\mathrm{d} \delta_{n}}$ is non-negative when $0\leq \delta_{n} \leq F_n/\Delta t_m$. 
	If $\Upsilon\geq 0$, we have $\frac{\mathrm{d} J_n}{\mathrm{d} \delta_n}>=0$. Thus, the optimal solution is achieved when $\delta_{n}=0$ for all $n=1,\cdots,m-1$, which contradicts (\ref{coro_P1}a). Hence, we should have $\Upsilon<0$. Due to $\frac{\mathrm{d}^2 J_n}{\mathrm{d} (\delta_{n})^2}=\frac{d^2 E_n}{d (\delta_{n})^2}>0$ in the region of $[0,F_n/\Delta t_m]$, we can obtain that $\frac{\mathrm{d} J_n}{\mathrm{d} \delta_n}$ monotonously increases.   Moreover, we have $\frac{\mathrm{d} J_n}{\mathrm{d} \delta_{n}}|_{\delta_{n}=0}=\Upsilon<0$.
	Therefore, we consider the following two cases.
	
	\textit{Case 1: $\frac{\mathrm{d} J_n}{\mathrm{d} \delta_{n}}|_{\delta_{n}=F_n/\Delta t_m}\leq0$, i.e., $\Upsilon\leq -\frac{3\kappa \Delta t_{m+1}F_n^2}{(\Delta t_m+\Delta t_{m+1})\Delta t_m}$.} In this case, $\frac{\mathrm{d} J_n}{\mathrm{d} \delta_{n}}<=0$ in the region of $0\leq \delta_{n} \leq F_n/\Delta t_m$. Therefore, the optimal solution is $\delta_{n}^*=F_n/\Delta t_m$. 
	
	\textit{Case 2: $\frac{\mathrm{d} J_n}{\mathrm{d} \delta_{n}}|_{\delta_{n}=F_n/\Delta t_m}>0$, i.e., $\Upsilon> -\frac{3\kappa \Delta t_{m+1}F_n^2}{(\Delta t_m+\Delta t_{m+1})\Delta t_m}$.} In this case, $\frac{\mathrm{d} J_n}{\mathrm{d} \delta_{n}}$ has a null point in the region of $0\leq \delta_{n} \leq F_n/\Delta t_m$. Thus, $J_n$ decreases first and then increases. Through solving $\frac{\mathrm{d} J_n}{\mathrm{d} \delta_{n}}=0$, we obtain that
 \begin{small}
	\begin{align} \label{coro1_eq2}
		\delta_{n}^*=\left\{\begin{aligned}
			&\frac{-F_n+\sqrt{F_n^2-\Upsilon\Xi(\Delta t_{m+1}-\Delta t_m)}}{\Delta t_{m+1}-\Delta t_m}, \text{if }\Delta t_{m}<\Delta t_{m+1}\\
            &\quad\quad\quad\quad\quad\quad\quad\quad\quad\quad\quad\quad\quad\quad\quad\ \ \text{ or } \Delta t_m>\Delta t_{m+1},\\
			&-\frac{\Upsilon\Xi}{2F_n}, \quad\quad\quad\quad\quad\quad\quad\quad\quad\quad\quad\quad \  \text{if }\Delta t_{m}=\Delta t_{m+1},
		\end{aligned} \right.
	\end{align}\end{small}\\
 where $\Xi=\frac{(\Delta t_m+\Delta t_{m+1})^2}{3\kappa \Delta t_m\Delta t_{m+1}}$. Meanwhile, the optimal $\Upsilon^*$ should satisfy constraint (\ref{coro_P1}a). 
	Obviously, both the above two cases satisfy $\delta^*_n>0$ $(n=1,\cdots,m-1)$, completing the proof.  \hfill$\blacksquare$
}

	\section{Proof of Lemma   \ref{lemma1}}  \label{proof_lemma1}
\setcounter{equation}{0}
\renewcommand\theequation{G.\arabic{equation}} 
	 \begin{figure} 
	\centering 
	\subfigure[$\Delta t_m>\Delta t_{m+1}$]{\label{}
		\includegraphics[height=1.5in]{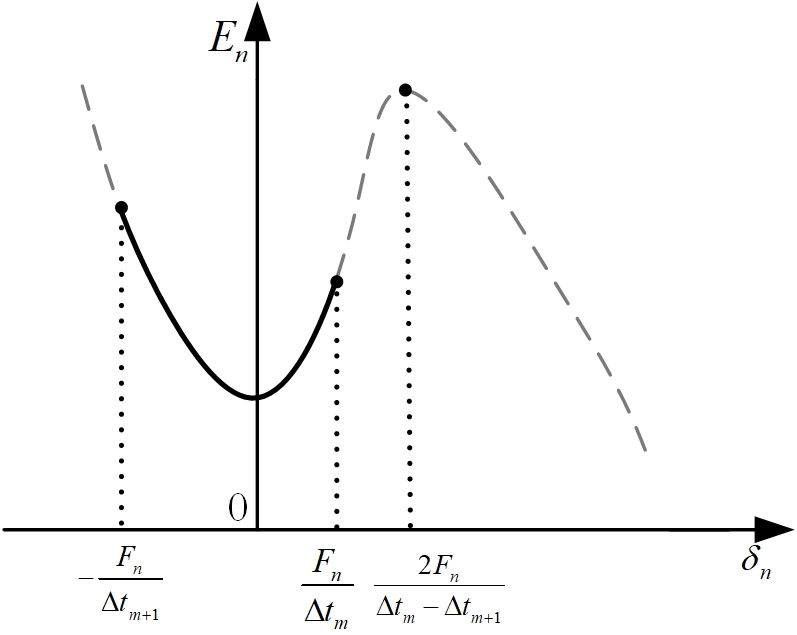}}
	\subfigure[$\Delta t_m<\Delta t_{m+1}$]{\label{}
		\includegraphics[height=1.5in]{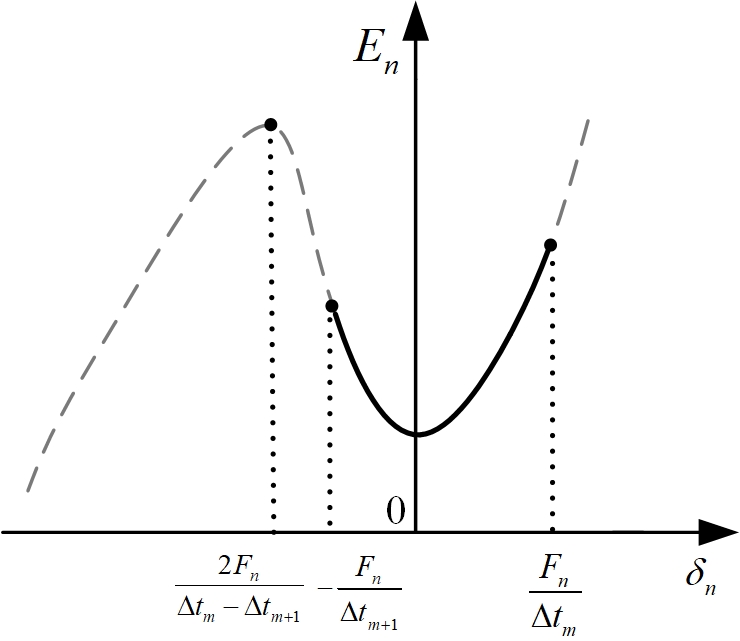}}
	\subfigure[$\Delta t_m=\Delta t_{m+1}$]{\label{}
		\includegraphics[height=1.5in]{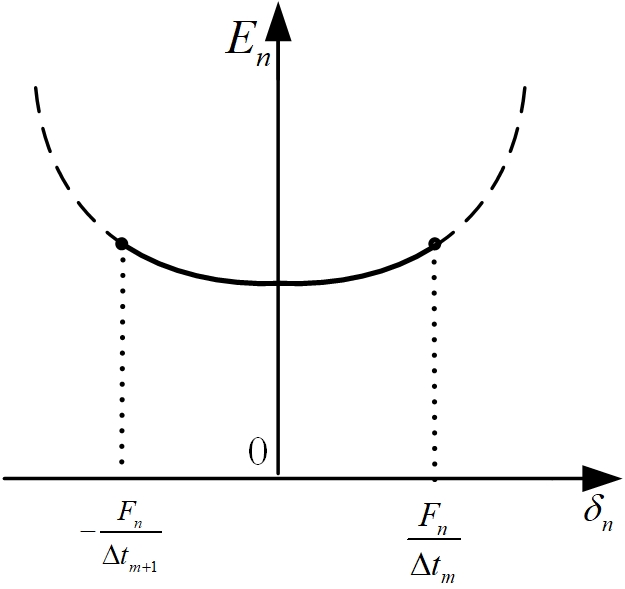}}
	\caption{Illustrations of three cases for the analysis of energy consumption in Appendix~\ref{proof_lemma1}.}
	\label{En_fig}\vspace{-3ex}
\end{figure}

Denote $\delta_n= f_{n,m}- f_{n,m+1}$. Since $ f_{n,m}$ and $f_{n,m+1}$ should be larger than or equal to zero, we can deduce that $-F_n/\Delta t_{m+1}\leq\delta_n\leq F_n/\Delta t_m$. Hence, we can obtain that
\begin{align}
	f_{n,m}=\frac{F_n+\delta_n \Delta t_{m+1}}{\Delta t_m+\Delta t_{m+1}},\ 
	f_{n,m+1}=\frac{F_n-\delta_n \Delta t_{m}}{\Delta t_m+\Delta t_{m+1}}.
\end{align}
Therefore, the energy consumption in time slot $t_m$ and $t_{m+1}$ can be given by
\begin{align} \label{lemma1_eq1.2}
	E_n&=\kappa f_{n,m}^3\Delta t_m+\kappa f_{n,m+1}^3\Delta t_{m+1},\nonumber\\
	&=\kappa\frac{(F_n+\delta_n\Delta t_{m+1})^3}{(\Delta t_m+\Delta t_{m+1})^3}\Delta t_m+\kappa\frac{(F_n-\delta_n\Delta t_{m})^3}{(\Delta t_m+\Delta t_{m+1})^3}\Delta t_{m+1}. 
\end{align}
Taking the first derivative of $E_n$ with respect to $\delta_n$, we have
\begin{align}\label{lemma1_eq2}
	\frac{\mathrm{d}E_n}{\mathrm{d}\delta_n}
	&=\frac{3\kappa \Delta t_m\Delta t_{m+1}}{(\Delta t_m+\Delta t_{m+1})^2}\delta_n\left(\delta_n(\Delta t_{m+1}-\Delta t_m)+2F_n\right). 
\end{align}
Equation \eqref{lemma1_eq2} has two null points: $0$ and $\frac{2F_n}{\Delta t_m-\Delta t_{m+1}}$. We consider the following three cases. 

\textit{Case 1: $\Delta t_m>\Delta t_{m+1}$.}   In this case, we have $0<\frac{2F_n}{\Delta t_m-\Delta t_{m+1}}$. The energy consumption decreases when $\delta_n<0$ and $\delta_n>\frac{2F_n}{\Delta t_m-\Delta t_{m+1}}$ while increases when $0\leq\delta_n\leq\frac{2F_n}{\Delta t_m-\Delta t_{m+1}}$, as shown in Fig.~\ref{En_fig}(a). Since we can easily prove that $\frac{2F_n}{\Delta t_m-\Delta t_{m+1}}>F_n/\Delta t_m$, the minimum energy consumption is obtained when $\delta_n=0$, i.e., $f^*_{n,m}=f^*_{n,m+1}$.  

\textit{Case 2: $\Delta t_m<\Delta t_{m+1}$.}   In this case, we have $\frac{2F_n}{\Delta t_m-\Delta t_{m+1}}<0$. The energy consumption increases when  $\delta_n<\frac{2F_n}{\Delta t_m-\Delta t_{m+1}}$ and $\delta_n>0$ while decreases when $\frac{2F_n}{\Delta t_m-\Delta t_{m+1}}\leq\delta_n\leq0$, as shown in Fig.~\ref{En_fig}(b). Similarly, since we can prove that $\frac{2F_n}{\Delta t_m-\Delta t_{m+1}}<-F_n/\Delta t_{m+1}$, the minimum energy consumption is obtained when $\delta_n=0$.  

\textit{Case 3: $\Delta t_m=\Delta t_{m+1}$.}   In this case, two null points  coincide, i.e., $\frac{2F_n}{\Delta t_m-\Delta t_{m+1}}=0$. Therefore, the energy consumption decreases when  $\delta_n<0$ while increases when $\delta_n\geq0$, as shown in Fig.~\ref{En_fig}(c). $\delta_n=0$ is the solution that minimizes energy consumption. 

In summary, the energy consumption when $f^*_{n,m}=f^*_{n,m+1}=\frac{F_n}{\Delta t_m+\Delta t_{m+1}}$ is the most energy efficient solution.  \hfill$\blacksquare$

	\section{Proof of Proposition   \ref{proposition_t}}  \label{proof_proposition_t}
\setcounter{equation}{0}
\renewcommand\theequation{H.\arabic{equation}} 	
	According to Theorem~\ref{theorem_5}, we have $x_{i-1,i}>0$ $(\forall i \in \{2,\cdots,K+1\})$. Hence, according to \eqref{Lag_i}, we have $\xi^* - \sum_{n=i+1}^K\rho^*_n h_{\pi_n} \eta P_0  -\omega^*_iF_{\max}=\sum_{n=1}^{i-1}2\kappa\frac{(x_{n,i})^3}{(\Delta t^*_i)^3}+\rho^*_i\frac{2\lambda(A_{\pi_i})^3}{h_{\pi_i}(\Delta t^*_{i})^{3}}>0$. Therefore, it can be derived that $\Delta t^*_{i}=\sqrt[3]{\frac{2\kappa h_{\pi_i}\sum_{n=1}^{i-1}(x_{n,i})^3+2\rho^*_{i}\lambda (A_{\pi_i})^3}{\xi^*h_{\pi_i} - \sum_{n=i+1}^K\rho^*_n (h_{\pi_n})^2 \eta P_0  -\omega^*_ih_{\pi_i}F_{\max}}}$ $(\forall i \in \{2,\cdots,K-1\})$. Similarly, based on \eqref{Lag_K} and \eqref{Lag_K+1}, we have  $\Delta t^*_K=\sqrt[3]{\frac{2\kappa h_{\pi_K}\sum_{n=1}^{K-1}(x_{n,K})^3+2\rho_K^*\lambda(A_{\pi_K})^3}{\xi^*h_{\pi_K}-\omega^*_Kh_{\pi_K}F_{\max}}}$ and 
$\Delta t_{K+1}^*=\sqrt[3]{\frac{2\kappa \sum_{n=1}^K(x_{n,K+1})^3}{\xi^*-\omega^*_{K+1}F_{\max}}}$.

Besides, according to \eqref{Lag_K+1}, we have $\xi^*=\sum_{n=1}^{K}2\kappa\frac{(x_{n,K+1})^3}{(\Delta t^*_{K+1})^3} +\omega^*_{K+1}F_{\max}>0$ since $\sum_{n=1}^K(x_{n,K+1})^3>(x_{K,K+1})^3=(A_{\pi_K}I_{\pi_K})^3>0$. Therefore, we can obtain that $\sum_{i=0}^{K+1}\Delta t^*_i=T$. Furthermore, based on \eqref{Lag_1}, we have $\xi^*=\sum_{n=1}^K\rho^*_nh_{\pi_n}\eta P_0$. Due to that $\xi^*$ is positive, there exists at least an $n\in\mathcal{K}$ such that $\rho_n^*>0$. This indicates that for energy causality constraints (\ref{P2}a), at least a device is run out of energy after offloading, i.e., this device uses all the harvested energy for transmission. Substituting \eqref{Lag_0} into \eqref{Lag_1}, we have $\rho^*_1\frac{2\lambda(A_{\pi_1})^3}{h_{\pi_1}(\Delta t^*_{1})^{3}} = \rho^*_1 h_{\pi_1} \eta P_0$. If $\rho_1^*>0$, we can deduce that $\Delta t^*_1=\sqrt[3]{\frac{2\lambda}{(h_{\pi_1})^2\eta P_0}}A_{\pi_1}$. Moreover, since $\frac{\lambda(A_{\pi_1})^3}{h_{\pi_1}(\Delta t_{1})^{2}}= \Delta t_0h_{\pi_1}\eta P_0$, we have $\Delta t_0^*=\sqrt[3]{\frac{\lambda}{4(h_{\pi_1})^2\eta P_0}}A_{\pi_1}$. If $\rho^*_1=0$,  $\frac{\mathrm{d}\mathcal{L}}{\mathrm{d}\Delta t^*_1}=0$ can be  guaranteed for arbitrary $\Delta t_1$ satisfying $\frac{\lambda(A_{\pi_1})^3}{h_{\pi_1}(\Delta t_{1})^{2}}\leq \Delta t_0h_{\pi_1}\eta P_0$. 
That means any pairs of $\Delta t_0$ and $\Delta t_1$ satisfying $\frac{\lambda(A_{\pi_1})^3}{h_{\pi_1}(\Delta t_{1})^{2}}\leq \Delta t_0h_{\pi_1}\eta P_0$ and $\Delta t_0+\Delta t_1=T-\sum_{i=2}^{K+1}\Delta t^*_i$ are the optimal solutions. 
Hence, according to $\Delta t_0+\Delta t_1=T-\sum_{i=2}^{K+1}\Delta t_i$, we should have $\frac{\lambda(A_{\pi_1})^3}{h_{\pi_1}(\Delta t_{1})^{2}}\leq \left(T-\sum_{i=2}^{K+1}\Delta t_i-\Delta t_1\right)h_{\pi_1}\eta P_0$, i.e., $\Psi(\Delta t_1)\triangleq\left(T-\sum_{i=2}^{K+1}\Delta t_i-\Delta t_1\right)(\Delta t_1)^2(h_{\pi_1})^2\eta P_0-\lambda(A_{\pi_1})^3\geq0$. Taking the derivative of $\Psi(\Delta t_1)$, we have $\Psi'(\Delta t_1)=\Delta t_1\left(2T-2\sum_{i=2}^{K+1}\Delta t_i-3\Delta t_1\right)(h_{\pi_1})^2\eta P_0$, which has two null points $\Delta t_1=0$ and $\Delta t_1=\frac{2}{3}\left(T-\sum_{i=2}^{K+1}\Delta t_i\right)$. Thus, $\Psi(\Delta t_1)$ increases in the region of $\Big(0,\frac{2}{3}\left(T-\sum_{i=2}^{K+1}\Delta t_i\right)\Big]$ and decreases in $\left(\frac{2}{3}\left(T-\sum_{i=2}^{K+1}\Delta t_i\right),\left(T-\sum_{i=2}^{K+1}\Delta t_i\right)\right)$. Additionally, we can obtain that $\Psi(0)=\Psi\left(T-\sum_{i=2}^{K+1}\Delta t_i\right)$. Therefore, we can choose $\Delta t^*_1$ as the unique null point of $\Psi(\Delta t_1)$ in the range of $\Delta t_1\in\Big(0,\frac{2}{3}\left(T-\sum_{i=2}^{K+1}\Delta t_i\right)\Big]$ without loss of generality. 
\hfill$\blacksquare$

	\bibliography{IEEEabrv,Ref}

\end{document}